\documentclass[aps,prl,reprint,superscriptaddress,nofootinbib,longbibliography]{revtex4-2}
\usepackage{mehdis}
\usepackage{amsmath,amssymb,amsthm,mathtools}
\usepackage{bm}
\usepackage{xcolor}
\usepackage{physics}
\usepackage{tikz}
\usetikzlibrary{arrows.meta,decorations.markings}
\usepackage[colorlinks=true,citecolor=blue,linkcolor=blue,urlcolor=blue]{hyperref}
\theoremstyle{plain} 
\newtheorem{theorem}{Theorem}
\newtheorem{corollary}{Corollary}[theorem]
\newtheorem{lemma}{Lemma}
\newtheorem{appendixcorollary}{Corollary}[section]

\newcommand{\dist}{\operatorname{dist}}
\newcommand{\poly}{\operatorname{poly}}
\newcommand{\supp}{\operatorname{supp}}
\newcommand{\cmi}{I(A\,{:}\,C{|}B)}
\newcommand{\rhozero}{\rho_0}
\newcommand{\noise}{\mathcal{N}}

\newcommand{\frakd}{\mathsf{D}}
\newcommand{\sweight}{\mathsf{w}}
\newcommand{\noiseweight}{\mathsf{s}}
\newcommand{\PD}{\mathsf{P}}

\begin{document}

\title{Stability of Quantum LDPC States under Local Pauli Noise}

\author{Tai-Hsuan Yang}
\affiliation{Department of Physics and Institute for Condensed Matter Theory, University of Illinois Urbana-Champaign, Urbana, Illinois 61801, USA}

\author{Andrew D. Kim}
\affiliation{Department of Physics and Institute for Condensed Matter Theory, University of Illinois Urbana-Champaign, Urbana, Illinois 61801, USA}

\author{Akash Vijay}
\affiliation{Department of Physics and Institute for Condensed Matter Theory, University of Illinois Urbana-Champaign, Urbana, Illinois 61801, USA}

\author{Jong Yeon Lee}
\affiliation{Department of Physics and Institute for Condensed Matter Theory, University of Illinois Urbana-Champaign, Urbana, Illinois 61801, USA}
\affiliation{Korea Institute for Advanced Study, Seoul 02455, South Korea} 

\date{\today}
\preprint{PRL manuscript}

\begin{abstract}
Quantum low-density parity-check (qLDPC) codes have a canonical mixed state---the normalized projector onto the code space---that lies at the intersection of quantum error correction and quantum many-body physics, yet the stability of its information-theoretic structure under local noise remains poorly understood. Recent progress on mixed-state phases of matter has identified the spatial decay of conditional mutual information (CMI) $I(A:C|B)$ as a key diagnostic of such stability, analogous to the role of the spectral gap for quantum ground states.
We prove that, for the normalized code-projector state of a $\mathbb{Z}_d$ qLDPC stabilizer code under local Pauli noise, the CMI across a buffer $B$ separating $A$ from $C$ decays exponentially with the width of $B$ in two complementary noise regimes independent of system size. 
At weak noise, this holds for full-system tripartitions, where correlations between the two sides require an error cluster to cross the buffer, an exponentially unlikely event.
At strong noise, exponential decay holds whenever the buffer supports no logical operator, since correlations can persist only along rare chains of buffer qudits that remain informative.
Under an additional space-uniformity assumption, we establish the robustness of Levin--Wen topological entropy.
On finite-dimensional lattices of size $n$, the weak-noise bound further implies a constant-layer quasi-local Petz recovery channel of range $O(\poly(\log n)/\varepsilon)$ that recovers the noiseless state within trace-distance error $\varepsilon$.
Both CMI bounds apply to general qLDPC interaction graphs, including geometries where such constant-depth quasi-local recovery need not exist. 
Thus, the exponential decay of the noise-induced CMI provides a general criterion for mixed-state phase stability.
\end{abstract}
\maketitle

\emph{Introduction.---}
Quantum low-density parity-check (qLDPC) codes are built from sparse stabilizer constraints, with each check acting on a bounded number of qudits and each qudit participating in a bounded number of checks. 
These local constraints collectively protect logical information encoded nonlocally across the system~\cite{qLDPC2021PRXQ}. 
In geometrically local realizations, such as quantum-double models and related topological stabilizer codes, this nonlocal structure is reflected in topological entanglement entropy, information-convex structures, and entanglement-based probes of anyon data~\cite{Kitaev1997,Kitaev2005,KitaevPreskill_2006,LevinWen2006,Shi2019ICS,Shi2019fusion,BombinQuantumDouble2008}.
More generally, qLDPC code states can realize nontrivial quantum phases of matter even beyond conventional spatially local geometries. Recent works have made this perspective precise by establishing their stability under perturbations that are local on the qLDPC interaction graph~\cite{Yinldpccodeasphaseofmatter2025,Deldpccodeasphaseofmatter2025,perturbativestabilityldpccodeli2025}.

Environmental noise poses a complementary stability question: whether qLDPC code states remain stable as mixed-state phases of matter under local quantum channels. 
The challenge of characterizing stable mixed-state phases has motivated several approaches, including quantum-memory and entanglement diagnostics~\cite{lee2024exact,Muller2024,MullerCI2025,kim2024errorthresholdsykcodes,lee2024mixedstatetopologicalordercoherent,wang2025analogtopologicalentanglemententropy,howmuchentanglementisneeded,chen2023separability,wang2023intrinsic, vijay2026hierarchyrenyicoherentinformation, temkin2025chargeinformedquantumerrorcorrection, vijay2026informationcriticalphasesdecoherence}, doubled-state and symmetry-based constructions~\cite{lee2022symmetry,Lee_2023,lessa2024strongtoweakspontaneoussymmetrybreaking,ma2024symmetryprotectedtopologicalphases,Ma2023,ma2025topologicalphasesaveragesymmetries}, operator-algebraic and categorical structures~\cite{Ellison2025,Negari2024,Ogata2025}, local reversibility~\cite{hsieh2026locallystable,chen2026notestrongquantummarkov, Sang2024,mixedstatelocalrev2025}, and mixed-state entanglement bootstrap~\cite{yang2025topologicalmixedstatesphases}.
A central quantity in these frameworks is the conditional mutual information (CMI)~\cite{ Sang2024, yang2025topologicalmixedstatesphases,mixedstatelocalrev2025} 
\[
  I(A:C|B)_{\rho} = S(AB)_{\rho}+S(BC)_{\rho} - S(B)_{\rho}-S(ABC)_{\rho},
\]
defined for a state $\rho$ and a tripartition $ABC$ in which $B$ separates $A$ from $C$, as illustrated in Fig.~\ref{fig:cmi-partitions}. Small CMI implies an approximate quantum Markov property: the state on $ABC$ can be reconstructed, with controlled error, from an appropriate marginal by a recovery channel acting locally near the buffer~\cite{Fawzi2014,Junge2015,PETZ_2003,pinchedpetz2016}. Together with two-way connectivity under local quantum channels, the coarse-grained decay of CMI---either along the time direction of a Lindbladian evolution realizing the channel~\cite{Sang2024, mixedstatelocalrev2025} or across a spatial domain wall~\cite{yang2025topologicalmixedstatesphases}---provides a rigorous framework for organizing mixed states into equivalence classes of phases.


\begin{figure}[!t]
\centering
\begin{tikzpicture}[every node/.style={font=\scriptsize}]
  \begin{scope}
    \fill[red!18] (1.30,0) circle (1.15);
    \fill[gray!22] (1.30,0) circle (0.75);
    \fill[blue!18] (1.30,0) circle (0.40);
    \draw (1.30,0) circle (1.15);
    \draw (1.30,0) circle (0.75);
    \draw (1.30,0) circle (0.40);
    \node at (1.30,0) {$A$};
    \node at (1.30,0.57) {$B$};
    \node at (1.30,0.96) {$C$};
    \node at (1.30,-1.48) {(a)};
  \end{scope}
  \begin{scope}[xshift=4.1cm]
    \fill[red!18] (0,0) circle (1.15);
    \fill[gray!22] (0:0.40) arc (0:180:0.40) -- (180:0.75) arc (180:0:0.75) -- cycle;
    \fill[yellow!28] (180:0.40) arc (180:360:0.40) -- (360:0.75) arc (360:180:0.75) -- cycle;
    \fill[blue!18] (0,0) circle (0.40);
    \draw (0,0) circle (1.15);
    \draw (0,0) circle (0.75);
    \draw (0,0) circle (0.40);
    \draw (0:0.40) -- (0:0.75);
    \draw (180:0.40) -- (180:0.75);
    \node at (0,0) {$A$};
    \node at (90:0.58) {$B$};
    \node at (270:0.58) {$D$};
    \node at (0:0.96) {$C$};
    \node at (0,-1.48) {(b)};
  \end{scope}
  \begin{scope}[xshift=6.9cm]
    \fill[blue!18] (35:0.45) arc (35:145:0.45) -- (145:1.15) arc (145:35:1.15) -- cycle;
    \fill[gray!22] (145:0.45) arc (145:215:0.45) -- (215:1.15) arc (215:145:1.15) -- cycle;
    \fill[red!18] (215:0.45) arc (215:325:0.45) -- (325:1.15) arc (325:215:1.15) -- cycle;
    \fill[gray!22] (325:0.45) arc (325:395:0.45) -- (395:1.15) arc (395:325:1.15) -- cycle;
    \draw (0,0) circle (1.15);
    \draw (0,0) circle (0.45);
    \draw (35:0.45) -- (35:1.15);
    \draw (145:0.45) -- (145:1.15);
    \draw (215:0.45) -- (215:1.15);
    \draw (325:0.45) -- (325:1.15);
    \node at (90:0.82) {$A$};
    \node at (180:0.82) {$B$};
    \node at (270:0.82) {$C$};
    \node at (0:0.82) {$B$};
    \node at (0,-1.48) {(c)};
  \end{scope}
\end{tikzpicture}
\vspace{-10pt}
\caption{ {\bf Representative CMI partitions where $B$ separates $A$ and $C$}.  (a) The disk-annulus tripartition for the Markov length.  (b) The CMI partition for space-uniformity. (c) The Levin--Wen partition for topological entropy. For higher dimensions, one replace disks in the figure by balls.}
\vspace{-12pt}
\label{fig:cmi-partitions}
\end{figure}
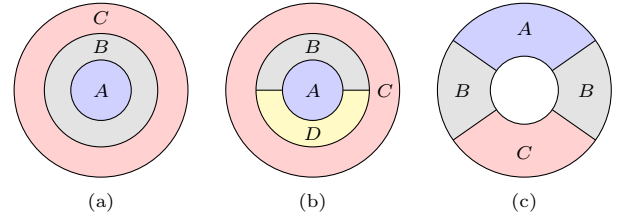

In this work, we study the stability of the qLDPC code state defined by the normalized projector onto the code space, $\rho_0={P_{\mathcal C}}/{\Tr P_{\mathcal C}}$~\cite{misc1}. Maximally mixed logical degrees of freedom removes correlations specific to an encoded state while preserving the stabilizer structure; this allows us to isolate correlations generated by noise-induced stabilizer excitations.
We show that, under local Pauli noise, the CMI across a separating buffer decays exponentially with the buffer width, both for sufficiently weak and for sufficiently strong noise, under different assumptions. 
In the weak-noise regime, this holds for every full-system tripartition $A\sqcup B\sqcup C=\Lambda$ with $B$ separating $A$ from $C$, for which the reference CMI vanishes automatically. For tripartitions of a subsystem $A\sqcup B\sqcup C\subsetneq\Lambda$, we provide a counterexample showing that the LDPC conditions alone are insufficient, and establish stability under an additional space-uniformity assumption.
In the strong-noise regime, provided an additional local-normalizer completeness condition holds, we again get exponential decay of the CMI. In both regimes, the decay rate is independent of system size and controlled by the noise strength and bounded local parameters.



This provides an analytic explanation for previous numerical observations~\cite{Sang2024, yang2025topologicalmixedstatesphases} and establishes the stability of mixed-state fixed points under Pauli noise~\cite{yang2025topologicalmixedstatesphases}, paralleling corresponding stability results for pure-state phases~\cite{Hasting_quasiadiabaticcontinuation_2005,Deldpccodeasphaseofmatter2025,perturbativestabilityldpccodeli2025}. As an application, we establish rigorous regimes where the CMI decays exponentially for the $\mathbb{Z}_{2}$ Toric Code under $X$ dephasing, see Fig.~\ref{fig:Toric_Code}.
It further yields a quasi-local recovery scheme in which the recovery maps can be organized into a constant number of parallel layers of quasi-local Petz maps~\cite{Sang2024}. 
Related works have established exponential CMI decay for noisy finite-temperature code states in the disk--annulus geometry and, at high temperature, for arbitrary tripartitions using pinning and cluster-expansion methods~\cite{ZhangGopalakrishnan2025,zhang2025stabilitymixedstatephasesweak,Kuwaharagibbsstate2020}.


\emph{Setup.---}
Consider qudits of local dimension $d$ placed on a finite set of sites $\Lambda$.  
On each site $i \in \Lambda$, we define generalized Pauli operators $X_i,Z_i$ satisfying $Z_iX_i=\omega X_iZ_i$, where $\omega=e^{2\pi i/d}$. 
Let $\mathcal I$ be a finite index set and $\{S_a\}_{a\in\mathcal I}$ be a family of stabilizer generators, and let $\mathcal S=\langle S_a:a\in\mathcal I\rangle$ be the Abelian stabilizer group that they generate.
Our setting is restricted to qLDPC codes with each generator $S_a$ acting nontrivially on at most $\sweight$ qudits and each qudit contained in the support of at most $\frakd$ generators.
Whenever we refer to a ``check'', we always mean one of these local qLDPC stabilizer generators.
We define the qLDPC interaction graph by taking the qudits as vertices and connecting two qudits whenever they appear in a common stabilizer check. The distance between two qudits is then given by the shortest-path distance on this graph.
For qudit regions $R,R', B\subseteq\Lambda$, $\dist(R,R')_B$ is the minimum graph distance between a qudit in $R$ and a qudit in $R'$ through some buffer region $B$.

Throughout, the regions $A,B,C$ are subsets of qudits. When $B$ separates $A$ from $C$, we write $\partial B_A$ for the $A$-side boundary adjacent to $B$, namely $\partial B_A = \{ i \in A| \dist(i, B) = 1\}$.
Similarly, $\partial B_C$ denotes the corresponding $C$-side boundary. 
For a region $R$, the relevant stabilizer subgroup is
\begin{equation}
    \mathcal S(R):=\{s\in\mathcal S:\supp(s)\subseteq R\},
\end{equation}
which may be strictly larger than the subgroup generated by checks that are individually supported in $R$. 
This is because products of checks that extend outside $R$ may cancel their support on $\bar R$, leaving a stabilizer supported entirely within $R$.
For any state $\sigma$ and region $R\subseteq\Lambda$, we write $\sigma_R:=\operatorname{Tr}_{\bar R}\sigma$.
We take the reference state $\rho_0$ to be the normalized code-projector state:
\begin{equation} \label{eq:rho_def}
  P_{\mathcal {C}} = \frac{1}{|\mathcal{S}|}\sum_{g \in \mathcal S} g, \quad 
  \rho_0 = \frac{P_{\mathcal{C}}}{\text{Tr}(P_{\mathcal{C}})} = \frac{1}{d^{|\Lambda|}}\sum_{g \in \mathcal {S}} g.
\end{equation}
We consider noisy states $\rho=\noise(\rho_0)$ generated by independent onsite Pauli channels,
\begin{equation}
 \noise=\bigotimes_{i\in\Lambda}\noise_i,\quad
 \noise_i(\cdot)=\sum_{u\in G}q_i(u)P_i(u)(\cdot) P_i(u)^\dagger,
\end{equation}
where $G:=\mathbb Z_d^2$ and $P_i(a,b)\,{:=}\,X_i^aZ_i^b$. Uniform depolarizing noise corresponds to $q_i(0)=1-p$ and $q_i(u)=p/(d^2-1)$ for $u\neq0$.

For an error configuration $e=(e_i)_{i\in\Lambda}\in G^\Lambda$, let
$P(e):=\prod_iP_i(e_i)$ be the corresponding Pauli operator. We denote the image of the stabilizer group $\mathcal S$ under this Pauli-label representation by $\mathsf S\subseteq G^\Lambda$. For a label $x$ and region $Q$, $x|_Q$ denotes the restriction of $x$ to $Q$. 
We define the commutation character of $P(e)$
$\chi_e:\mathcal S\to\mathbb Z_d$ by
\begin{equation}
\label{eq:character}
 gP(e)=\omega^{\chi_e(g)}P(e)g,
 \qquad g\in\mathcal S.
\end{equation}
Let $E$ denote the random Pauli error configuration drawn from the noise distribution. For a region $R\subseteq\Lambda$, we define the corresponding stabilizer-syndrome random variable $\widehat\Sigma_R:=\chi_E|_{\mathcal S(R)}$, which records the commutation phases of the random error with all stabilizers supported in $R$.

\begin{figure}[!t]
    \centering
    \includegraphics[width=0.98\columnwidth]{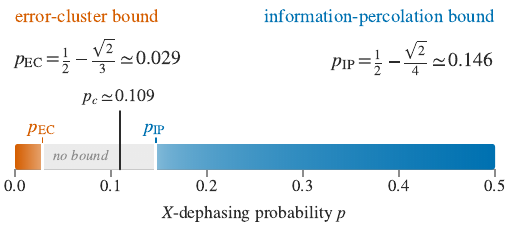}
    \caption{\textbf{Rigorous exponential CMI decay of the toric code.}  Under $X$-dephasing with probability $p$, the error-cluster bound (orange) applies for $p<p_{\mathrm{EC}}$, where $p_{\mathrm{EC}}$ solves $6\sqrt{p(1-p)}=1$. The information-percolation bound (blue) applies for $p>p_{\mathrm{IP}}$, where $p_{\mathrm{IP}}$ solves $(1-2p)^2=1/2$, i.e., where the contraction coefficient $(1-2p)^2$ of a noisy link reaches the square-lattice bond-percolation threshold. The two bounds leave an intermediate grey region uncovered, which includes the Nishimori critical point $p_c\simeq0.109$.}
    \label{fig:Toric_Code}
    \vspace{-12pt}
\end{figure}


\emph{Main results.---}
The normalized code-projector state $\rho_0$ is exactly Markov. For a full tripartition $A\sqcup B\sqcup C=\Lambda$ in which $B$ separates $A$ from $C$, as in Fig.~\ref{fig:cmi-partitions}(a), no stabilizer generator can intersect both $A$ and $C$. Every stabilizer can therefore be generated from stabilizers supported on $AB$ and $BC$, implying $\cmi_{\rho_0}=0$. 

Local Pauli noise generally generates a nonzero CMI, and the central question is whether these induced correlations remain short-ranged. We show that the CMI decays exponentially with $\dist(A,C)_B$ in two complementary noise regimes, governed by distinct mechanisms. At weak noise, the \emph{error-cluster bound} (Theorem~\ref{thm:cmi-stability}) uses the sparsity of errors: correlations between $A$ and $C$ require an error cluster to cross the buffer, which is exponentially unlikely. At strong noise, the \emph{information-percolation bound} (Theorem~\ref{thm:information_percolation}) relies on the complementary effect that the buffer retains very little information. Correlations between $A$ and $C$ can then persist only through a connected chain of buffer qudits that remain informative, whose probability is again exponentially suppressed.

These two regimes lie on opposite sides of the decoding transition, near which the Markov length is expected to diverge. For the $\mathbb Z_2$ toric code under $X$-dephasing, our bounds apply for $p<0.029$ and $p>0.146$, respectively, bracketing the Nishimori critical point $p_c\simeq0.109$.


\begin{theorem}[Error-Cluster Bound]
\label{thm:cmi-stability}
Let $\rhozero$ be the normalized code-projector state of a $\mathbb{Z}_d$ qLDPC code, as defined in Eq.~\eqref{eq:rho_def}, and let $\rho=\noise(\rhozero)$ for an onsite Pauli noise channel $\noise$.
For every tripartition $A\sqcup B\sqcup C=\Lambda$ such that $B$ separates $A$ from $C$, define
\begin{equation}
  \alpha:=\mathrm e\frakd\sweight
  \max_{i\in\Lambda}
  \left(
    \sum_{g\in G\setminus\{0\}}\sum_{u\in G}
    \sqrt{q_i(u)q_i(u+g)}
  \right).
  \label{eq:onsite-alpha}
\end{equation}
where $\mathrm e$ is Euler's number. If $0<\alpha<0.2$ (the upper bound is not tight), then
\begin{align}   \label{eq:main-bound}
I(A:C|B)_{\rho}
   &\leq  f( |\partial B_A|, |\partial B_C| )  e^{-c_0\dist(A,C)_B}, \nonumber \\
   f(x,y) &:= 3e^{c_0}\frakd^2  (x+y)^2 \log d
\end{align}
where $c_0=\frac12\log(1/\alpha)$ is system-size independent and
$\dist(A,C)_B-1$ is the minimum number of internal $B$ vertices in a crossing component. 
For uniform depolarizing noise,
$q(0)=1-p$ and $q(u)=p/(d^2-1)$ for $u\neq0$, and Eq.~\eqref{eq:onsite-alpha}
reduces to
\begin{equation}
  \alpha=\mathrm e\frakd\sweight
  \left(2\sqrt{p(1-p)(d^2-1)}+(d^2-2)p\right).
\end{equation}
\end{theorem}

\begin{proof}
Here we provide a sketch of the proof, with details given in \appref{main-proof}.
First, we use the data processing inequality to remove the noise acting on $A$ and $C$:
\begin{align}
    I(A:C|B)_{ {\cal N}(\rho_0)} \leq I(A:C|B)_{{\cal N}_B(\rho_0)},
\end{align}
where $\noise_B$ denotes the product noise channel restricted to $B$. After this, any possible dependence on the sizes of $A$ and $C$ is reduced to the sizes of $\partial B_A$ and $\partial B_C$.

From Eq.~\eqref{eq:character}, each error realization selects a stabilizer sector determined by its syndrome, while the logical block remains maximally mixed.
For a qudit region $R \subseteq \Lambda$, averaging over errors therefore produces an orthogonal mixture of stabilizer sectors labeled by $\widehat\Sigma_R$, all with the same maximally mixed logical block.
It follows that the change in entropy on $R$ under the noise channel is entirely determined by the stabilizer-syndrome distribution induced by the errors:
\begin{align}
S\left((\mathcal N_B(\rho_0))_R\right)-S\left((\rho_0)_R\right)&=H(\widehat\Sigma_R),
\label{eq:regional-syndrome-entropy-increment}
\end{align}
where $H$ denotes the Shannon entropy.

To proceed, we introduce the following syndrome variables. $\Sigma_{\textrm{int},\partial B_A}$ records the syndromes of ($i$) checks crossing the $A$--$B$ interface and ($ii$) $B$-supported stabilizers generated by products of checks on the $A$ side. We define $\Sigma_{\textrm{int},\partial B_C}$ analogously. By contrast, $\Sigma_{\textrm{bulk},B}$ records the syndromes of checks supported entirely within $B$, while $\widehat\Sigma_B$ may additionally contain syndromes of possibly nonlocal $B$-supported stabilizers generated through cancellations of checks across either interface. With these definitions, the noisy-state CMI obeys
\begin{align}
     &\cmi_{\mathcal N_B(\rho_0)} \nonumber \\
 &=
 H(\widehat\Sigma_{AB})+H(\widehat\Sigma_{BC})
 -H(\widehat\Sigma_B)-H(\widehat\Sigma_{ABC}) \nonumber
 \\
 &=
 I(\widehat\Sigma_{AB}:\widehat\Sigma_{BC}
   |\widehat\Sigma_B)  = I (\Sigma_{\textrm{bulk},A} : \Sigma_{\textrm{bulk},C} | \widehat\Sigma_{B})   \nonumber \\
 &\leq I(\Sigma_{\textrm{int}, \partial B_A}:
   \Sigma_{\textrm{int}, \partial B_C}|\Sigma_{\textrm{bulk}, B}). 
\label{eq:regional-syndrome-cmi-difference}
\end{align}
The first equality follows from Eq.~\eqref{eq:regional-syndrome-entropy-increment}.
The second follows because $\widehat\Sigma_{AB}$ and $\widehat\Sigma_{BC}$ share the common restriction $\widehat\Sigma_B$, while their union determines $\widehat\Sigma_{ABC}$.
For the third equality, note that the noise acts only on $B$. Consequently, after conditioning on $\widehat\Sigma_B$, the information in $\widehat\Sigma_{AB}$ that is not already contained in $B$ is precisely $\Sigma_{\textrm{bulk},A}$, and similarly for the $C$-side.
Finally, $\widehat\Sigma_B$ can be decomposed into $\Sigma_{\textrm{bulk},B}$ together with $B$-supported stabilizers generated by checks extending from $A$ or $C$. These additional stabilizers are associated with the two interfaces and can therefore be moved from the conditioning variable to the corresponding arguments of the conditional mutual information. The chain rule for CMI then implies that this operation can only increase the CMI, yielding the final inequality. 

We remark that the condition that $A\sqcup B\sqcup C = \Lambda$ is important. First, it implies $I(A:C|B)_{\rho_0} =0$; second, it guarantees that there is no non-local check across the region $ABC$ from tracing out the complement. A more detailed explanation is given in \appref{main-proof}.

Instead of evaluating the entropy terms in the last line of Eq.~\eqref{eq:regional-syndrome-cmi-difference} directly, we first bound the total variation distance $d_{\rm TV}$ between the conditional distribution $\PD(\Sigma_{\textrm{int},\partial B_A},\Sigma_{\textrm{int},\partial B_C}|\Sigma_{\textrm{bulk},B}=\tau)$ and the product distribution $\PD(\Sigma_{\textrm{int},\partial B_A}|\Sigma_{\textrm{bulk},B}=\tau)\PD(\Sigma_{\textrm{int},\partial B_C}|\Sigma_{\textrm{bulk},B}=\tau)$.
Then, we use the bound on $d_{\rm TV}$ in an entropy continuity bound~\cite{Audenaert_2007, Winter_2016} on the mutual information for each conditioned syndrome $\tau$.

We bound $d_{\rm TV}$ by constructing an explicit coupling between the conditional joint distribution and the product of its marginals. 
Specifically, we consider $\overline{\PD}_\tau(X,Y)$ such that $(i)$ its marginals obey 
\begin{align}
    \overline{\PD}_\tau(X) &= \PD(\Sigma_{\textrm{int},\partial B_A},\Sigma_{\textrm{int},\partial B_C} | \Sigma_{\textrm{bulk},B} = \tau), \nonumber \\
    \overline{\PD}_\tau(Y)&=\PD(\Sigma_{\textrm{int},\partial B_A} |\Sigma_{\textrm{bulk},B} =  \tau) \nonumber \\
    &\quad\cdot \PD(\Sigma_{\textrm{int},\partial B_C} | \Sigma_{\textrm{bulk},B} = \tau),
\end{align}
and $(ii)$ $X\,{=}\,Y$ unless the difference between the underlying qudit error configurations $e$ contains a connected component with support on both $\partial B_A$ and $\partial B_C$.
The syndrome distribution $\overline{\PD}_\tau(X)$ is determined by the underlying error distribution $\PD(e)$ via the check matrices determined by the stabilizers~\cite{qLDPC2021PRXQ}.  
The remaining technical challenge for the proof is constructing $\overline{\PD}_\tau(X,Y)$ to satisfy $(ii)$, which is detailed in \appref{main-proof}.


After constructing $\overline{\PD}_\tau(X,Y)$, we utilize the \emph{coupling lemma}~\cite{couplinglemma}, which states that the total variation distance $d_{\rm TV}$ of the two coupled distributions is bounded as 
\begin{align}
    d_{\rm TV} \leq \sum_{x \neq y} \overline{\PD}_\tau(X=x,Y=y),
\end{align}
where, by construction, the RHS is determined by the number of error configurations with at least one connected component supported on both $\partial B_A$ and $\partial B_C$. 
The LDPC assumptions bound the number of such connected components of length $\ell$ by a factor exponential in $\ell$, while the probability of each component is suppressed by a factor of order $p^\ell$.  
For $p<p_{\mathrm{th}}$ with $p_{\mathrm{th}}$ depending only on the LDPC constants, the probabilistic suppression dominates the combinatorial growth.  
Summing over all penetrating connected components gives an exponentially small bound in $\dist(A,C)_B$.

The technique used to prove Theorem~\ref{thm:cmi-stability} extends naturally to geometrically $k$-local Pauli noise. 
See \appref{proof-k-local-corollary} for details.
\end{proof}

At strong noise, error clusters are no longer dilute. 
We instead bound the information that can propagate across the noisy buffer. 
To do so, we first express the quantum CMI in terms of a classical CMI in which the noise acts through independent onsite Pauli channels.


By the data processing inequality, we may restrict the noise to $B$ and consider $\widetilde\rho=\mathcal N_B(\rho_0)$. 
To represent the CMI of $\mathcal N_B(\rho_0)$ classically, let $E$ denote the Pauli-error label supported in $B$, and introduce an independent random label $\mathsf G$ sampled uniformly from the Pauli normalizer $\mathsf N=\mathsf S^\perp$ independently of $E$. 
We then define the site-resolved classical record $Y=\mathsf G+E$. 
Adding $\mathsf G$ preserves the stabilizer syndromes while making the restricted record $Y_Q:=Y|_Q$ uniform over all Pauli labels consistent with the syndrome in $Q$. 
This gives $S(\widetilde\rho_Q)=H(Y_Q)-|Q|\log d$. These volume terms cancel in the CMI, yielding $I(A:C\mid B)_{\widetilde\rho}=I(Y_A:Y_C\mid Y_B)$.

Under the local-normalizer completeness assumption of Theorem \ref{thm:information_percolation}, $\mathsf G|_B$ can be constructed from independent uniform coefficients $h_\mu$ associated with local normalizer generators $t_\mu$.
The CMI is then bounded by the conditional correlations between the coefficients of generators touching $A$ and those touching $C$. Each buffer qudit $i\in B$ provides a noisy observation of the local coefficients acting on it,
\begin{equation}
    Y_i=F_i(h)+E_i,\quad \,
    F_i(h):= \hspace{-15pt} \sum_{\mu:\,i\in\supp(t_\mu)} \hspace{-15pt} h_\mu(t_\mu)_i.
\end{equation}
Although the coefficients $h_\mu$ are initially independent, conditioning on overlapping observations can correlate them across the buffer. Strong noise suppresses these correlations by obscuring the observed combinations.
Thus, at strong noise, the relevant small parameter is not $p$, but the information retained by each noisy observation.

To quantify this information loss, consider
the additive channel $X\mapsto X+E_i$ with $E_i\sim q_i$.
Its strong-data-processing coefficient $\eta_i$ is the smallest constant satisfying
\begin{equation}
    I(U:X+E_i)\leq \eta_i I(U:X)
    \label{eq:lnhn-eta}
\end{equation}
for every joint distribution over an auxiliary variable $U$ and an input label $X\in G$, with $E_i$ independent of $(U,X)$. Thus, $\eta_i$ measures the largest fraction of mutual information that this channel can preserve and is small in the strong noise regime. 

Finally, the information-percolation inequality of Polyanskiy and
Wu~\cite{PolyanskiyWu2020} then converts this local contraction of mutual information into a bound on the classical CMI through a percolation process.
In this process, each buffer qudit $i$ is independently open with probability $\eta_i$, and an open qudit connects the coefficients of the generators acting on it. The CMI is bounded by $\log d$ times the expected number of generators at one interface connected to the other through open qudits. Since each generator acts on a bounded number of qudits and each qudit participates in a bounded number of generators, the number of connecting paths grows at most exponentially with their length. If the opening probabilities suppress long paths faster than their number grows, the resulting CMI bound decays exponentially with the distance across the buffer.

\begin{theorem}[Information-Percolation Bound]
\label{thm:information_percolation}
Let $\rho=\noise(\rhozero)$ for an onsite Pauli noise channel $\mathcal{N}$, and let $A,B,C\subseteq\Lambda$ be pairwise disjoint, nonempty regions such that
\begin{equation}
    I(A:C\mid B)_{\rho_0}=0.
    \label{eq:lnhn-markov}
\end{equation}
Assume that there exists a subgroup $\mathsf{H}$ with $\mathsf S \subseteq \mathsf H \subseteq \mathsf N$, such that:
\begin{enumerate}
    \item $\mathsf H$ admits a generating set $\{t_\mu\}$, with each generator acting on at most $w_{\mathsf H}$ qudits and each qudit participating in at most $D_{\mathsf H}$ generators.
    \item The generators capture all normalizer degrees of freedom visible in the buffer $B$, i.e. $\mathsf H|_B=\mathsf N|_B$.
\end{enumerate}
Let $\Gamma_{\mathsf H}$ be the interaction graph induced by the generators $\{ t_\mu \}$, and suppose $B$ separates $A$ from $C$ in this graph. We denote by $\dist_{\mathsf H}(A,C)_B$ the length of the shortest $A$-to-$C$ path whose internal vertices lie in $B$, and by $\partial_{\mathsf H} A$ and $\partial_{\mathsf H} C$ the corresponding interfaces adjacent to $B$. 
Setting $\bar\eta:=\max_{i\in B}\eta_i$ and $\alpha_{\mathsf H}:=(w_{\mathsf H}-1)(D_{\mathsf H}-1)\bar\eta$, if $0<\alpha_{\mathsf H}<1$, we find that
\begin{equation}
 \begin{split}
 I(A:C\mid B)_\rho\leq{}&
 D_{\mathsf H}(\log d)\min\{|\partial_{\mathsf H} A|,|\partial_{\mathsf H} C|\}\\
 &\times e^{-c_{\mathsf H}[\dist_{\mathsf H}(A,C)_B-1]}.
 \end{split}
 \label{eq:lnhn-spatial}
\end{equation}
where $c_{\mathsf H}:=\log(1/\alpha_{\mathsf H})>0$ is independent of system size. The proof is given in Appendix~\ref{app:lnhn-proof}.
\end{theorem}

\paragraph{Counterexample for arbitrary tripartitions.}
We have established the stability of CMI for full-system tripartitions $A\sqcup B\sqcup C=\Lambda$, with $B$ separating $A$ from $C$.
For tripartitions of a subregion, however, the CMI need not remain stable under noise even when $B$ separates $A$ from $C$ and the code has LDPC structure.
We give the following counterexample in two dimensions illustrated in Fig.~\ref{fig:translation-counterexample}: place a qubit at each vertex of a finite square lattice and take the translation-invariant stabilizers
\begin{equation}
 S_{x,y}:=
 X_{x,y}X_{x+1,y}X_{x,y+1}X_{x+1,y+1}.
 \label{eq:translation-counterexample-check}
\end{equation}
Every qubit belongs to four weight-four checks. Take a subsystem with $H\geq3$ and $L\geq2$ as shown in Fig.~\ref{fig:translation-counterexample} and choose
\begin{equation}
\begin{aligned}
 A_L&=\{(0,y):0\leq y\leq H\}
      \cup\{(x,H):1\leq x\leq L\},\\
 B_L&=\{(x,0):1\leq x\leq L\},\\
 C_L&=\{(L,1),(L+1,0),(L+1,1)\}.
\end{aligned}
\end{equation}
Then $\dist(A_L,C_L)_{B_L}=L+O(1)$ and 
\begin{equation}
 T_L:=\prod_{x=0}^{L-1}\prod_{y=0}^{H-1}S_{x,y}
 =X_{0,0}X_{L,0}X_{0,H}X_{L,H}
 \in\mathcal S(A_LB_L)
\end{equation}
while $\mathcal S_{B_LC_L}:=\{\mathbbm{1}, S_{L,0}\}$ and $\mathcal S(B_L)=\{\mathbbm{1}\}$.

For every $1\leq j\leq L$, the rectangular product
\begin{equation}
 T_j:=\prod_{x=0}^{j-1}\prod_{y=0}^{H-1}S_{x,y}
 =X_{0,0}X_{j,0}X_{0,H}X_{j,H}
 \in\mathcal S(A_LB_L)
\end{equation}
has three vertices in $A_L$ and only $(j,0)$ in $B_L$.
Apply independent $Z$ dephasing with probability $p$ on $B_L$.
The syndrome of $T_j$ is the error bit at $(j,0)$, so the
$A_LB_L$ syndrome determines all $L$ independent buffer-error bits.
The $B_LC_L$ syndrome records only the bit at $q=(L,0)$ through
$S_{L,0}$, while the $B_L$ syndrome is trivial.
The $A_LB_LC_L$ syndrome also determines the entire buffer-error
string and contains no additional randomness. The entropy increments
on $A_LB_L$, $B_LC_L$, $B_L$, and $A_LB_LC_L$ are therefore
$Lh_2(p)$, $h_2(p)$, $0$, and $Lh_2(p)$, respectively.
Thus
\begin{equation}
 I(A_L:C_L|B_L)_{\mathcal N(\rho_0)}
 -I(A_L:C_L|B_L)_{\rho_0}=h_2(p),
 \label{eq:translation-counterexample-cmi}
\end{equation}
independently of $L$.  For noise on the entire region $ABC$, the system-size independent CMI increasement by the noise on $B$-side still have an $L$-independent contribution, see \appref{full-dephasing-counterexample} for exact calculations.

\begin{figure}[t]
\centering
\begin{tikzpicture}[
    x=0.4cm,
    y=0.4cm,
    every node/.style={font=\scriptsize},
    q/.style={circle,draw=black,fill=white,inner sep=1.25pt},
    Aq/.style={circle,draw=red!70!black,fill=red!20,inner sep=1.6pt},
    Bq/.style={circle,draw=black,fill=gray!35,inner sep=1.6pt},
    Cq/.style={circle,draw=blue!70!black,fill=blue!20,inner sep=1.6pt}
]

  \def\panelboxheight{4}

  \foreach \panel/\length/\panelheight/\xoff in {
    a/2/3/0,
    b/8/4/8
  }{

    \pgfmathsetmacro{\yoff}{(\panelboxheight-\panelheight)/2}

    \begin{scope}[shift={(\xoff,\yoff)}]

      \pgfmathtruncatemacro{\rightedge}{\length+1}

      \draw[step=1,gray!45,very thin]
        (0,0) grid (\rightedge,\panelheight);

      \draw[orange!80!black,dashed,line width=1pt]
        (0,0) rectangle (\length,\panelheight);

      \draw[blue!70!black,densely dashed,line width=1pt]
        (\length,0) rectangle (\rightedge,1);

      \foreach \x in {0,...,\rightedge}{
        \foreach \y in {0,...,\panelheight}{
          \node[q] at (\x,\y) {};
        }
      }

      \foreach \y in {0,...,\panelheight}{
        \node[Aq] at (0,\y) {};
      }

      \foreach \x in {1,...,\length}{
        \node[Aq] at (\x,\panelheight) {};
      }

      \foreach \x in {1,...,\length}{
        \node[Bq] at (\x,0) {};
      }

      \node[Cq] at (\length,1) {};
      \node[Cq] at (\rightedge,0) {};
      \node[Cq] at (\rightedge,1) {};

      \foreach \sx/\sy in {0/0,\length/0,0/\panelheight,\length/\panelheight}{
        \draw[orange!80!black,line width=0.9pt] (\sx,\sy) circle[radius=0.14];
      }
      \foreach \sx/\sy in {\length/1,\rightedge/0,\rightedge/1}{
        \draw[blue!70!black,line width=0.8pt]
          ({\sx-0.20},{\sy-0.20}) rectangle ({\sx+0.20},{\sy+0.20});
      }
      \draw[black,thin] (\length,-0.23) -- ({\length+0.45},-0.60)
        node[right,fill=white,inner sep=1pt] {$q$};

      \draw[{Latex}-{Latex}]
        (0,-1.05) --
        node[below] {$L=\length$}
        (\length,-1.05);

      \draw[{Latex}-{Latex}]
        (-0.75,0) --
        node[left] {$H$}
        (-0.75,\panelheight);

      \node[red!70!black]
        at (0,{\panelheight+0.65}) {$A_L$};

      \node
        at ({\length/2},-0.55) {$B_L$};

      \node[blue!70!black]
        at ({\rightedge+0.65},0.55) {$C_L$};

    \end{scope}

    \node[
      anchor=south west,
      font=\bfseries
    ] at ({\xoff-2.7},{\panelboxheight+0.3}) {(\panel)};
  }

  \node at (13,5.0) {$D=\Lambda\setminus A_LB_LC_L$};

\end{tikzpicture}
\vspace{-8pt}
\caption{ {\bf Counterexample.} In general, CMI stability is not guaranteed for an arbitrary tripartition of the form in Fig.~\ref{fig:cmi-partitions}(b). 
(a) A small instance with $L=2$ and $H=3$. (b) A larger instance with $L=8$ and $H=4$. White vertices belong to the unobserved region $D$; the displayed patches are embedded in the square torus. 
Orange rings mark the four-site support of $T_L$. The orange dashed rectangle encloses the plaquettes multiplied to obtain $T_L$.}
\label{fig:translation-counterexample}
\vspace{-12pt}
\end{figure}
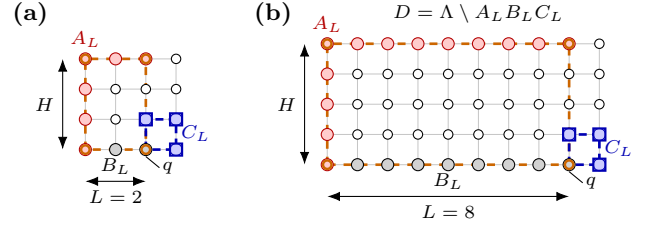

For subsystem tripartitions, products of local checks can cancel their support outside $ABC$, leaving stabilizers in $\mathcal S(AB)$ and $\mathcal S(BC)$ that overlap on a small region $Q\subseteq B$. An error on $Q$ can affect both side syndromes while remaining invisible to $\mathcal S(B)$, generating an $O(1)$ conditional correlation independently of the buffer width. The qLDPC conditions alone do not exclude this mechanism, so CMI stability for arbitrary tripartitions requires additional assumptions.

\begin{corollary}[Stability of the space-uniformity CMI]
\label{cor:space-uniformity-stability}
Let $\rho_0$ be a $\mathbb Z_d$ qLDPC stabilizer code state satisfying $I(A:C|B)_{\rho_0}=0$ for every partition having the same topology as Fig.~\ref{fig:cmi-partitions}(b).
Then there exists a system-size-independent weak-noise threshold such that, for sufficiently weak onsite Pauli noise $\rho=\mathcal N(\rho_0)$,
\begin{align} \label{eq:space-uniformity-stability}
I(A:C|B)_\rho \leq K\bigl(|\partial B_A|+|\partial B_C|\bigr)^2 e^{-c\operatorname{dist}(A,C)_B},
\end{align}
for every partition of the same topology, where $K,c>0$ are independent of the system size.
\end{corollary}

\begin{proof}
As in the proof of Theorem~\ref{thm:cmi-stability}, data processing reduces the problem to $\mathcal N_B(\rho_0)$.
The zero-CMI condition is equivalent to
\begin{align} \label{eq:space-uniformity-markov-factorization}
\mathcal S(ABC)=\left\langle\mathcal S(AB),\mathcal S(BC)\right\rangle.
\end{align}
Applying this identity recursively to an ordered cover of $B$ by overlapping local patches $Q_1,\ldots,Q_N$ gives
\begin{align}
\mathcal S(B)=\left\langle\mathcal S(Q_1),\ldots,\mathcal S(Q_N)\right\rangle,
\end{align}
so $\mathcal S(B)$ admits a locally supported generating set; see Appendix~\ref{sec:space-uniformity-local-generation}. Next, let $M_A:=N_{r_{\rm M}}^B(\partial B_A)$ be a fixed-width neighborhood of the $AB$ interface. Applying the zero-CMI condition to the corresponding Fig.~\ref{fig:cmi-partitions}(b)-type partitions $(A,B\setminus M_A, M_A)$ gives
\begin{align}
    \mathcal S(AB)=\left\langle\mathcal S(AM_A),\mathcal S(B)\right\rangle
\end{align}
and similarly we get $\mathcal S(BC)=\left\langle\mathcal S(CM_C),\mathcal S(B)\right\rangle$. 
Hence
\begin{align}\label{eq:space-uniformity-classical-reduction}
I(A:C|B)_{\mathcal N_B(\rho_0)}
=
I(\widehat\Sigma_{AM_A}:\widehat\Sigma_{CM_C}|\widehat\Sigma_B).
\end{align}
The coupling argument following Eq.~\eqref{eq:regional-syndrome-cmi-difference} now applies with the two interfaces replaced by $M_A$ and $M_C$.
Only connected components crossing from $M_A$ to $M_C$ contribute, and $\dist(M_A,M_C)\geq\dist(A,C)_B-O(1)$.
The crossing-component bound and entropy continuity then give Eq.~\eqref{eq:space-uniformity-stability}, with only a constant change in the boundary prefactor.
\end{proof}

One important consequence of is the robustness of the Levin--Wen topological entropy as a mixed-state phase invariant proposed in the mixed-state bootstrap program~\cite{yang2025topologicalmixedstatesphases}.
For the partition in Fig.~\ref{fig:cmi-partitions}(c), the Levin--Wen topological entropy is $\gamma_{\rm LW}(\rho):=\frac12 I(A:C|B)_\rho$~\cite{LevinWen2006}.
The stability follows from the space-uniformity bound by the spatial domain-wall deformation stated below.

\begin{corollary}[Stability of the Levin--Wen CMI]
\label{cor:levin-wen-stability}
Let $\rho_0$ satisfy $I(A:C| B)_{\rho_0}=0$ for every partition with the topology of Fig.~\ref{fig:cmi-partitions}(b), and let $\rho=\mathcal N(\rho_0)$ for onsite Pauli noise below the threshold of Corollary~\ref{cor:space-uniformity-stability}.
Consider a Levin--Wen partition as in Fig.~\ref{fig:cmi-partitions}(c): $A$ and $C$ are simply connected, $B=B_1\sqcup B_2$ consists of two disjoint simply connected regions that separate $A$ from $C$, and $ABC$ is not simply connected.
Let $r$ be the minimum buffer width appearing in a fixed local deformation of this partition and let $\delta(r)$ be the largest right-hand side of Eq.~\eqref{eq:space-uniformity-stability} among the Fig.~\ref{fig:cmi-partitions}(b)-type partitions encountered in that deformation.
Then
\begin{equation}
 \left|
 \gamma_{\rm LW}(\rho)-\gamma_{\rm LW}(\rho_0)
 \right|
 \leq C_{\rm def}\,\delta(r),
 \label{eq:levin-wen-stability}
\end{equation}
where $C_{\rm def}=O(1)$ is the number of local deformation steps.
In particular, the Levin--Wen CMI is stable up to an exponentially small correction in $r$.
\end{corollary}

\begin{proof}
Apply the noise outside a clean region containing the original Levin--Wen partition, and deform the partition into the noisy region. Apply the remaining noise in the clean region; this leaves the reduced state on the displaced partition unchanged. Finally, reverse the deformation in the fully noisy state. The space-uniformity bound controls every local move, including those crossing the spatial noise boundary, because partially applied onsite noise obeys the same threshold. Summing the outward and return errors proves Eq.~\eqref{eq:levin-wen-stability}. For details, see Appendix~\ref{sec:levin-wen-deformation}.
\end{proof}

\emph{Markov length and local recovery.---}
Theorem~\ref{thm:cmi-stability} (and the extension to $k$-local Pauli noise) shows that, below a finite weak-noise threshold, the CMI of the partition in Fig.~\ref{fig:cmi-partitions}(a) of the noisy state is exponentially decaying with the separating width.  With such CMI behavior on a finite-dimensional lattice, one can construct a constant number of layers of quasi-local Petz recovery maps that approximately reverse the noise~\cite{Sang2024}.

\begin{corollary}[Constant-depth quasi-local decoder]
\label{cor:constant-depth-block-petz}
Let $\rho_0$ be the normalized code state of a qLDPC code on a finite-dimensional lattice $\Lambda$, and let $\rho=\noise(\rho_0)$ for an onsite Pauli noise channel $\noise$ below the threshold of Theorem~\ref{thm:cmi-stability}.
Then there exists a constant-depth, range-$R=\Theta(\log(|\Lambda|/\varepsilon))$ channel $\mathcal D_R$ such that
\begin{align}
\norm{\mathcal D_R(\rho)-\rho_0}_1\leq\varepsilon.
\end{align}
Each layer of $\mathcal D_R$ consists of Petz recovery channels acting on non-overlapping blocks of $\Lambda$.
\end{corollary}

The result follows from Ref.~\cite{Sang2024}; a proof is given in \appref{proof-constant-depth-block-petz}. On a general interaction graph, the same construction need not remain constant-depth or quasi-local in support size.
Suppose the graph has volume and boundary growth
\begin{align}
|N_R(x)|\leq e^{h_VR},
\quad
|\partial N_R(x)|\leq e^{h_\partial R}.
\end{align}
The one-block argument of \appref{proof-constant-depth-block-petz} then gives recovery error $\epsilon \sim O\left(e^{-(c_0-2h_\partial)R/2}\right)$.
Thus, when $c_0>2h_\partial$, a graph radius $R=O(\log(|\Lambda|/\varepsilon))$ still suffices for recovery to error $\varepsilon$.
However, because an $R$-ball can contain $e^{O(h_VR)}$ qudits, both the support size of individual recovery maps and the number of layers needed to parallelize them can grow polynomially with $|\Lambda|/\varepsilon$.

\emph{Discussion.---}
We have established exponential suppression of noise-generated CMI in weak (Error-cluster bound) and strong (Information-percolation bound) noise regimes for qLDPC code states.
Together, these results provide a mixed-state analogue of the stability of pure-state topological order under weak local Hamiltonian perturbations~\cite{Bravyi2010,Bravyi_2011}.
From the quantum-error-correction perspective, Pauli-noise stability has traditionally been formulated through vanishing logical error under decoding~\cite{gottesman2014faulttolerantquantumcomputationconstant}.
Here the stability criterion is instead the spatial decay of CMI, which applies directly to general qLDPC interaction graphs. Exponential CMI decay therefore provides the more general information-theoretic stability criterion, with local reversibility as an operational consequence.
Several extensions are natural.
The decay constant $\alpha$ depends on both stabilizer weight and qudit degree; the degree dependence may be removable by counting paths through stabilizer labels~\cite{gottesman2014faulttolerantquantumcomputationconstant,Dumer_threshold}.
Other directions include weak non-Pauli noise and finite-depth local channels, as well as group-valued qudit codes, non-Abelian quantum doubles, and string-net models~\cite{Levin2005,Levingeneralizedstringnet,twistedquantumdouble2013,sala2024stability,sala2024D4decoherence,sala2026deocherence}.
The latter require replacing the Abelian syndrome-disagreement description by one incorporating non-Abelian charge labels, fusion multiplicities, and correlated cluster weights.

\section{acknowledgments}
We thank Sarang Gopalakrishnan, Yifan Zhang, Shengqi Sang, and colleagues for useful discussions.  THY is supported by Taiwan-UIUC fellowship program. A.V. and J.Y.L. are supported by the faculty startup grant at the University of Illinois, Urbana-Champaign, and the IBM-Illinois Discovery Accelerator Institute. We acknowledge the use of OpenAI ChatGPT to assist in the proof of the theorems. All mathematical arguments were independently verified, and the authors take full responsibility for the manuscript.

\bibliographystyle{apsrev4-2}
\bibliography{ref}

\clearpage
\onecolumngrid
\appendix
\setcounter{secnumdepth}{2}
\section{Preliminaries}\label{sec:preliminaries}

\begin{lemma}[Polymer Bound via Site Percolation]
\label{lem:lattice_animal}
Let $G$ be a graph with maximum degree $\Delta$. For a fixed vertex $u \in V(G)$, let $a_\ell(u)$ denote the number of connected subsets $S \subseteq V(G)$ of size $|S| = \ell$ that contain $u$. Then,
\begin{equation*}
    a_\ell(u) \le (\mathrm{e}(\Delta + 1))^\ell.
\end{equation*}
\end{lemma}

\begin{proof}
For any connected subset $S$, we define its closed neighborhood as $N[S] = S \cup \{v \in V(G) : \dist(v, S) = 1\}$. Since the maximum degree of $G$ is $\Delta$, each vertex $v \in S$ contributes itself and at most $\Delta$ neighbors. Thus, by a union bound, the size of the closed neighborhood is bounded by
\begin{equation}
    |N[S]| \le \sum_{v \in S} (1 + \deg v) \le (\Delta + 1)|S|.
\end{equation}

Fix a vertex $u \in V(G)$,
suppose each vertex in $G$ is independently declared ``open'' with probability $t \in (0,1)$ and ``closed'' with probability $1-t$. 
Let $C(u)$ denote the open cluster (the maximal connected open induced subgraph) containing $u$. 
For the cluster $C(u)$ to be exactly a specific connected subset $S$ of size $\ell$ containing $u$, two conditions must be met:
\begin{enumerate}
    \item All $\ell$ vertices in $S$ must be open, which occurs with probability $t^\ell$.
    \item All vertices in the external boundary of $S$, given by $N[S] \setminus S$, must be closed.
\end{enumerate}

Since the states of the vertices are independent, the probability that $C(u)$ is exactly $S$ is bounded below by
\begin{equation}
    \Pr(C(u) = S) \ge t^\ell (1 - t)^{\Delta \ell}.
\end{equation}

The events $\{C(u) = S\}$ for distinct connected subsets $S$ are mutually disjoint. Summing over all $a_\ell(u)$ valid subsets $S$ of size $\ell$ containing $u$, the total probability cannot exceed 1:
\begin{equation}
    \sum_{\substack{S \ni u \\ |S| = \ell}} \Pr(C(u) = S) \le 1.
\end{equation}

Substituting the lower bound into the sum yields
\begin{equation}
    a_\ell(u) t^\ell (1 - t)^{\Delta \ell} \le 1 \implies a_\ell(u) \le \frac{1}{t^\ell (1 - t)^{\Delta \ell}}.
\end{equation}

This inequality holds for any valid probability $t \in (0,1)$. 
To obtain the tightest bound, we need to maximize the denominator by taking the derivative and setting it to zero. 
The optimal probability is $t = 1/(\Delta + 1)$, and by substituting this optimal value, we find
\begin{equation}
    a_\ell(u) \le \left[ (\Delta + 1) \left( 1 + \frac{1}{\Delta} \right)^\Delta \right]^\ell.
\end{equation}

Using the standard limit inequality $\left(1 + 1/x\right)^x < \mathrm{e}$ for all $x > 0$, we conclude that
\begin{equation}
    a_\ell(u) \le (\mathrm{e}(\Delta + 1))^\ell.
\end{equation}
\end{proof}

\begin{lemma}[Coupling Lemma]
\label{lem:coupling}
Let $\mu$ and $\nu$ be two probability distributions on a countable state space $\Omega$.
A coupling of $\mu$ and $\nu$ is a pair of random variables $(X,Y)$ on the same probability space such that
\begin{align*}
    X \sim \mu, \qquad Y \sim \nu.
\end{align*}
Then, for every coupling $(X,Y)$ of $\mu$ and $\nu$,
\begin{align*}
    d_{\mathrm{TV}}(\mu,\nu) \leq \Pr(X \neq Y),
\end{align*}
where
\begin{align*}
    d_{\mathrm{TV}}(\mu,\nu) = \frac{1}{2}\sum_{x\in \Omega} |\mu(x)-\nu(x)| = 1-\sum_{x\in\Omega}\min\{\mu(x),\nu(x)\}.
\end{align*}
\end{lemma}

\begin{proof}
First, let $(X,Y)$ be any coupling of $\mu$ and $\nu$.
For any subset $A\subseteq \Omega$, we have
\begin{align}
    \mu(A)-\nu(A) = \Pr(X\in A)-\Pr(Y\in A).
\end{align}
On the event $\{X=Y\}$, the indicators
$\mathbf{1}_{\{X\in A\}}$ and $\mathbf{1}_{\{Y\in A\}}$ are equal.
Therefore,
\begin{align}
    \mu(A)-\nu(A) \leq \Pr(X\neq Y).
\end{align}
Taking the supremum over all $A\subseteq \Omega$ gives
\begin{align}
    d_{\mathrm{TV}}(\mu,\nu) \leq \Pr(X\neq Y).
\end{align}
\end{proof}

\section{Proof of Theorem 1} 
\label{sec:main-proof}
\subsection{Setup and notation}
Let $\Lambda$ be a finite lattice with qudits of local dimension $d$ placed on each site. 
Define
\begin{align}
  G:=\mathbb Z_d^2,
  \qquad
  \omega:=e^{2\pi \mathrm{i} / d},
\end{align}
and let $X_i,Z_i$ be single-site generalized Pauli operators for each $i \in \Lambda$ such that
\begin{align}
    Z_i X_i = \omega X_i Z_i.
\end{align}
From now on, when we write ``Pauli operator'' or just ``Pauli'', we always mean ``generalized Pauli operator''.
Each operator in $\mathcal{P}_d^{(i)}$, the set of single-site Pauli operators on site $i$, is defined via a unique $(a,b) \in G$:
\begin{align}
    P_i(a,b) = X_i^{a} Z_i^{b} \in \cP_d^{(i)}.
\end{align}
Any Pauli operator (with arbitrary support on the lattice) in $\mathcal{P}_d^{\otimes |\Lambda|}$ is just a tensor product of single-site Pauli operators.

Let $\mathcal{S}=\langle S_a\rangle$ be an Abelian stabilizer group generated by a set of Pauli operators $\{S_a\}\subset \mathcal{P}_d^{\otimes |\Lambda|}$. 
As stated in the main text, we consider qLDPC codes for which each generator acts nontrivially on at most $\sweight$ qudits and each qudit is contained in the support of at most $\frakd$ generators.
Distances are measured in the interaction graph induced by these generators; two qudits have distance one if they appear in a common stabilizer generator, and two stabilizer generators have distance one if their supports overlap on at least one qudit.
For qudit regions $R,R', B\subseteq\Lambda$, $\dist(R,R')_B$ is the minimum graph distance between a qudit in $R$ and a qudit in $R'$ through some buffer region $B$.

Let $A \sqcup B \sqcup C = \Lambda$ be a (disjoint) tripartition of qudits such that $B$ separates $A$ from $C$, i.e. $\dist(A,C)_B > 1$.
We will denote by $\partial B_A$ the $A$-side boundary of $B$, i.e.
\begin{align}
    \partial B_A := \{ i \in A| \dist(i, B) = 1\},
\end{align}
and denote by $\partial B_C$ the corresponding $C$-side boundary of $B$.  

A label $e_i=(a_i,b_i)\in G$ represents the Pauli error at site $i \in \Lambda$ with corresponding Pauli operator $X_i^{a_i} Z_i^{b_i}$.
An error configuration on $B$ is then fully specified by
\begin{align}
  e=(e_i)_{i\in B}\in G^{|B|},
  \qquad e_i=(a_i,b_i),
\end{align}
and we write the corresponding Pauli error as
\begin{align}
  P(e):=\prod_{i\in B}X_i^{a_i}Z_i^{b_i},
  \qquad
  \supp(e):=\{i\in B:e_i\neq 0\}, \qquad |e| = |\supp(e)|.
\end{align}

Throughout the proof, an \textit{error distribution} means a distribution of Pauli error configurations on qudits, while a \textit{syndrome distribution} means the induced distribution of stabilizer-check values.

For an error configuration $e \in G^{|B|}$, a subset (or ``component'') $Q\subseteq \supp(e)$ is called \textit{connected} if for any
$i,j\in Q$, there exists a sequence
\begin{align}
    i=i_1,i_2,\ldots,i_m=j
\end{align}
with $i_k\in Q$ such that each consecutive pair $i_k,i_{k+1}$ is adjacent in the qudit interaction graph. For general onsite Pauli noise, let $q_i$ be the one-site distribution introduced in the main text.  Depolarizing noise is the special case
$q_i(0)=1-p_i$ and $q_i(u)=p_i/(d^2-1)$ for $u\neq0$.
The product error distribution on $G^{|B|}$ is
\begin{equation}
  q(e)=\prod_{i\in B}q_i(e_i).
\end{equation}

\subsection{Reduction to classical probability}
Let $\rho_0$ be the normalized code-projector state defined in Eq.~\eqref{eq:rho_def} of the main text.
We consider a product noise channel $\noise$, which is the composition of single-site Pauli noise channels $\noise_i$ for each $i \in \Lambda$.
We begin by using the data processing inequality to remove the noise acting on $A$ and $C$:
\begin{align}
    I(A:C|B)_{ {\cal N}(\rho_0)} \leq I(A:C|B)_{{\cal N}_B(\rho_0)}, \label{eq:data_processing_ineq}
\end{align}
where $\noise_B$ denotes the product noise channel with its support restricted to $B$. 
Any possible dependence of the RHS of Eq.~\eqref{eq:data_processing_ineq} on the sizes of $A$ and $C$ is localized to the two boundaries $\partial B_A$ and $\partial B_C$. 

Because $A\sqcup B\sqcup C=\Lambda$ and $B$ separates $A$ from $C$, no
chosen generator intersects both $A$ and $C$.  Define
\begin{equation}
\begin{aligned}
 \mathcal G_0&:=\left\langle S_a:
  a\in\mathcal I,\ \supp(S_a)\subseteq B\right\rangle,\\
 \mathcal G_A&:=\left\langle S_a:
  a\in\mathcal I,\ \supp(S_a)\cap A\neq\varnothing\right\rangle,\\
 \mathcal G_C&:=\left\langle S_a:
  a\in\mathcal I,\ \supp(S_a)\cap C\neq\varnothing\right\rangle,
\end{aligned}
\end{equation}
and let
\begin{equation}
 \mathcal K_A:=\mathcal G_A\cap\mathcal S(B),
 \qquad
 \mathcal K_C:=\mathcal G_C\cap\mathcal S(B).
\end{equation}
The groups $\mathcal K_A$ and $\mathcal K_C$ contain the effective
$B$-supported stabilizers obtained when products of $A$- or $C$-side
generators cancel their support outside $B$.  The four regional
groups then factor as
\begin{equation}
\begin{aligned}
 \mathcal S(B)&=\mathcal K_A\mathcal G_0\mathcal K_C,
 &\mathcal S(AB)&=\mathcal G_A\mathcal G_0\mathcal K_C,\\
 \mathcal S(BC)&=\mathcal K_A\mathcal G_0\mathcal G_C,
 &\mathcal S(ABC)&=\mathcal G_A\mathcal G_0\mathcal G_C.
\end{aligned}
\label{eq:intrinsic-stabilizer-factorizations}
\end{equation}
Every regional stabilizer can be separated into an $A$-side product, a product of generators supported in $B$, and a $C$-side product;
membership in the relevant region requires the unwanted exterior support
to cancel.

Using the commutation character defined in the setup, for a subgroup
$\mathcal H\leq\mathcal S$ write
$\Sigma_{\mathcal H}:=\chi_E|_{\mathcal H}$ for the induced random syndrome
character, where $E$ is the Pauli error drawn from $q$.  Thus
$\widehat\Sigma_R=\Sigma_{\mathcal S(R)}$.

For these subgroups, use the abbreviations
\begin{equation}
\begin{aligned}
 \Sigma_{\textrm{bulk},A}&:=\Sigma_{\mathcal G_A},
 &\Sigma_{\textrm{bulk},B}&:=\Sigma_{\mathcal G_0},
 &\Sigma_{\textrm{bulk},C}&:=\Sigma_{\mathcal G_C}.
\end{aligned}
\end{equation}
and define the complete side registers by
\begin{equation}
\begin{aligned}
 \Sigma_{\textrm{int},\partial B_A}
 &:=(\Sigma_{\textrm{bulk},A},\Sigma_{\mathcal K_A}),\\
 \Sigma_{\textrm{int},\partial B_C}
 &:=(\Sigma_{\textrm{bulk},C},\Sigma_{\mathcal K_C}).
\end{aligned}
\end{equation}
The entries in these tuples are not treated as independent: they are
restrictions of the same character $\chi_E$ and hence obey the consistency
relations on subgroup intersections.
Choose independent $\mathbb Z_d$-valued coordinates for the two side
registers, and let $\Omega_A$ and $\Omega_C$ denote their respective
numbers of coordinates.  Since
$\mathcal K_A\leq\mathcal G_A$ and $\mathcal K_C\leq\mathcal G_C$, the
effective-stabilizer syndromes introduce no additional independent
coordinates beyond those of $\Sigma_{\textrm{bulk},A}$ and $\Sigma_{\textrm{bulk},C}$.  Moreover, because
the error is supported in $B$, only local generators crossing the
corresponding interface can have a nonzero side syndrome.  Therefore,
\begin{equation}
 \Omega_A\leq\frakd|\partial B_A|,
 \qquad
 \Omega_C\leq\frakd|\partial B_C|.
 \label{eq:boundary-stabilizer-to-qudit-size}
\end{equation}

Equation~\eqref{eq:intrinsic-stabilizer-factorizations} implies
\begin{equation}
\begin{aligned}
 \widehat\Sigma_B
 &\simeq
 (\Sigma_{\mathcal K_A},\Sigma_{\textrm{bulk},B},\Sigma_{\mathcal K_C}),\\
 \widehat\Sigma_{AB}
 &\simeq
 (\Sigma_{\textrm{bulk},A},\Sigma_{\textrm{bulk},B},\Sigma_{\mathcal K_C}),\\
 \widehat\Sigma_{BC}
 &\simeq
 (\Sigma_{\mathcal K_A},\Sigma_{\textrm{bulk},B},\Sigma_{\textrm{bulk},C}),\\
 \widehat\Sigma_{ABC}
 &\simeq
 (\Sigma_{\textrm{bulk},A},\Sigma_{\textrm{bulk},B},\Sigma_{\textrm{bulk},C}),
\end{aligned}
\label{eq:regional-syndrome-factorizations}
\end{equation}
where $\simeq$ means that the complete character and the displayed tuple
of jointly consistent restrictions determine each other.  Moreover,
$\widehat\Sigma_B$ is a deterministic restriction of both
$\widehat\Sigma_{AB}$ and $\widehat\Sigma_{BC}$.  For the full
tripartition $A\sqcup B\sqcup C=\Lambda$,
\begin{equation}
 \mathcal S(ABC)
 =\langle\mathcal S(AB),\mathcal S(BC)\rangle,
\end{equation}
so $\widehat\Sigma_{ABC}$ is determined by
$(\widehat\Sigma_{AB},\widehat\Sigma_{BC})$.
For any region $R$, Pauli noise therefore maps the reduced code-projector state into
mutually orthogonal, equal-rank sectors labelled by
$\widehat\Sigma_R$:
\begin{equation}
 (\mathcal N_B(\rho_0))_R
 =\bigoplus_{\tau}
 \Pr(\widehat\Sigma_R=\tau)\,\rho_R^{(\tau)},
\end{equation}
where every $\rho_R^{(\tau)}$ has the same entropy as $(\rho_0)_R$.
Consequently,
\begin{equation}
 S\!\left((\mathcal N_B(\rho_0))_R\right)
 -S\!\left((\rho_0)_R\right)
 =H(\widehat\Sigma_R).
\end{equation}
Combining the four entropy increments and using
Eq.~\eqref{eq:regional-syndrome-factorizations} gives
\begin{equation}
\begin{aligned}
 \Delta I
 &:={I(A:C| B)}_{\mathcal N_B(\rho_0)}
   -{I(A:C| B)}_{\rho_0}\\
 &=I(\widehat\Sigma_{AB}:\widehat\Sigma_{BC}
     |\widehat\Sigma_B)\\
 &=I(\Sigma_{\textrm{bulk},A}:\Sigma_{\textrm{bulk},C}|
       \Sigma_{\textrm{bulk},B},\Sigma_{\mathcal K_A},\Sigma_{\mathcal K_C})\\
 &\leq I(\Sigma_{\textrm{int},\partial B_A}:\Sigma_{\textrm{int},\partial B_C}
          |\Sigma_{\textrm{bulk},B}).
\end{aligned}
\label{eq:extra-cmi-syndrome}
\end{equation}

Using these coordinate representations, write
$\sigma_B(e)$, $\sigma(e)_{\partial B_A}$, and
$\sigma(e)_{\partial B_C}$ for the coordinate vectors of the bulk and side
registers.  The bulk-syndrome distribution and conditional error law are
\begin{align}
 \mathsf P_B(\tau)
 &:=\Pr_q(\Sigma_{\textrm{bulk},B}=\tau)
   =\sum_{e:\sigma_B(e)=\tau}q(e),\\
 q_B^\tau(e)
 &:=\Pr_q(E=e|\Sigma_{\textrm{bulk},B}=\tau)
   =\frac{q(e)\mathbf 1\{\sigma_B(e)=\tau\}}
          {\mathsf P_B(\tau)}.
 \label{eq:conditional-error-law}
\end{align}
The last term in Eq.~\eqref{eq:extra-cmi-syndrome} decomposes as
\begin{equation}
 I(\Sigma_{\textrm{int},\partial B_A}:\Sigma_{\textrm{int},\partial B_C}|\Sigma_{\textrm{bulk},B})
 =\mathbb E_{\tau\sim\mathsf P_B}
 I(\Sigma_{\textrm{int},\partial B_A}:\Sigma_{\textrm{int},\partial B_C}
   |\Sigma_{\textrm{bulk},B}=\tau).
\end{equation}
To obtain an upper bound on $I(\Sigma_{\textrm{int},\partial B_A}:\Sigma_{\textrm{int},\partial B_C}| \Sigma_{\textrm{bulk},B}=\tau)$, we will first bound the total variation distance between the conditional distribution $\mathsf{P}_{\tau}(\Sigma_{\textrm{int},\partial B_A}, \Sigma_{\textrm{int},\partial B_C})$ and the product of its marginals $\mathsf{P}_{\tau}(\Sigma_{\textrm{int},\partial B_A})\mathsf{P}_{\tau}(\Sigma_{\textrm{int},\partial B_C})$ by constructing a coupling between them.

\subsection{Coupling argument}
Fix a bulk syndrome $\tau\in\operatorname{im}(\sigma_B)$.
We consider two error configurations independently sampled from the conditional distribution
\begin{align}
  e^{(1)},e^{(2)}\sim q_{B}^\tau.
\end{align}
We want to construct a coupling $(X,Y) \sim \overline{\PD}_\tau(X,Y)$ between two random variables $X$ and $Y$, both of which store the syndromes of checks on the $A$--$B$ and $B$--$C$ boundaries.
The only hard constraint for a valid construction is that the coupled distribution $\overline{\PD}_\tau(X,Y)$ has marginals
\begin{align}
    \overline{\PD}_\tau(X) &= \PD(\Sigma_{\textrm{int},\partial B_A},\Sigma_{\textrm{int},\partial B_C} | \Sigma_{\textrm{bulk},B} = \tau), \nonumber \\
    \overline{\PD}_\tau(Y) &=\PD(\Sigma_{\textrm{int},\partial B_A} |\Sigma_{\textrm{bulk},B} =  \tau) \cdot \PD(\Sigma_{\textrm{int},\partial B_C} | \Sigma_{\textrm{bulk},B} = \tau).
\end{align}
Naively, we could construct such a coupling by letting
\begin{align}
    X = \left( (e^{(1)})_{\partial B_A}, (e^{(1)})_{\partial B_C}\right), \quad Y = \left( (e^{(1)})_{\partial B_A}, (e^{(2)})_{\partial B_C} \right),
\end{align}
which gives the correct marginals.
The problem with this choice is that it neglects correlations between syndromes on $\partial B_A$ and $\partial B_C$ which are introduced by conditioning on the bulk syndrome $\tau$.
The total probability weight $\Pr[X \not= Y]$ is then uncontrolled, and thus Lemma \ref{lem:coupling} gives only a loose upper bound on the total variation distance.

To refine this strategy, we consider the \textit{Pauli label disagreement} of the errors $e^{(1)}$ and $e^{(2)}$, defined as
\begin{equation}
  \Delta := e^{(2)}-e^{(1)}\; \left(\text{mod } d\right)\in G^{B}.
\end{equation}
Since $e^{(1)}$ and $e^{(2)}$ have the same bulk syndrome $\tau$ and the syndrome map $\sigma(\cdot)$ is a homomorphism, we have $\sigma_{B}(\Delta)=0$.
Let $R$ be the union of all connected components of $\supp(\Delta)$ which are adjacent to ${\partial B_C}$ but not to ${\partial B_A}$, and write $\Delta_R$ as the restriction of $\Delta$ to $R$, padded with zeros for all sites not in $R$.
We then define a \textit{switching map}:
\begin{equation}
  \mathcal T(e^{(1)},e^{(2)})
  =
  (\widetilde e^{(1)},\widetilde e^{(2)})
  :=
  (e^{(1)}+\Delta_R,\; e^{(2)}-\Delta_R).
\end{equation}
On sites in $R$, the map $\mathcal{T}$ simply swaps the two Pauli labels:
\begin{align}
  \widetilde e^{(1)}_i=e^{(2)}_i,
  \qquad
  \widetilde e^{(2)}_i=e^{(1)}_i,
  \qquad i\in R.
\end{align}
On sites not in $R$, the two configurations are unchanged. 
The switching map $\mathcal{T}$ is bijective as its inverse is just itself, and it preserves the total probability weight of the two sampled errors:
\begin{equation}
  q_i(e^{(1)}_i)q_i(e^{(2)}_i) = q_i(\widetilde e^{(1)}_i) q_i(\widetilde e^{(2)}_i).
\end{equation}
Denote by $\mathsf{P}_\tau(\Sigma_{\textrm{int},\partial B_A}, \Sigma_{\textrm{int},\partial B_C})$ the joint syndrome distribution at the boundaries conditioned on $\Sigma_{\textrm{bulk},B} = \tau$.
Recall that $\sigma(e)_R$ is the (deterministic) syndrome on a qudit region $R$ for a fixed error configuration $e$. 
In what follows, we often write $\PD_\tau(\sigma(e)_R)$ instead of the more explicit expression for the probability $\PD_\tau(\Sigma_R = \sigma(e)_R)$. 
We have the following factorization for the probability $\mathsf{P}_\tau(\sigma(\widetilde e^{(1)})_{\partial B_A}, \sigma(\widetilde e^{(2)})_{\partial B_C})$:
\begin{equation}
\begin{split}
  &\mathsf{P}_\tau(\sigma(\widetilde e^{(1)})_{\partial B_A},
  \sigma(\widetilde e^{(2)})_{\partial B_C})
  \\
  &=
  \sum_{\{e^{(1)}, e^{(2)}
  |
  (\sigma(e^{(1)}+\Delta_R)_{\partial B_A},
  \sigma(e^{(2)}-\Delta_R)_{\partial B_C})
  =
  (\sigma(\widetilde e^{(1)})_{\partial B_A},
  \sigma(\widetilde e^{(2)})_{\partial B_C})\}}
  q_B^\tau(e^{(1)})q_B^\tau(e^{(2)})
   \\
   &=
  \sum_{\{\hat e^{(1)}, \hat e^{(2)}
  |
  (\sigma(\hat e^{(1)})_{\partial B_A},
  \sigma(\hat e^{(2)})_{\partial B_C})
  =
  (\sigma(\widetilde e^{(1)})_{\partial B_A},
  \sigma(\widetilde e^{(2)})_{\partial B_C})\}}
  q_B^\tau(\hat e^{(1)})q_B^\tau(\hat e^{(2)})
  \\
  &=
  \mathsf{P}_\tau(\sigma(\widetilde e^{(1)})_{\partial B_A})
  \mathsf{P}_\tau(\sigma(\widetilde e^{(2)})_{\partial B_C}).
\end{split}
\end{equation}
Let $\mathcal C_{\tau}$ be the set of error configuration pairs (with bulk syndrome $\tau$) for which $\supp(\Delta)$ contains a connected component adjacent to both $\partial B_A$ and $\partial B_C$.
More explicitly,
\begin{align}
  \mathcal C_{\tau}
  :=
  \left\{
  (e^{(1)},e^{(2)}) : \exists \; \Gamma \subseteq \supp(e^{(2)} - e^{(1)}) \text{ s.t. } \Gamma \text{ is connected}, \; \dist(\Gamma,\partial B_A) = 1, \; \dist(\Gamma,\partial B_C) = 1 
  \right\}.
\end{align}
Now, consider the Pauli label disagreement $\widetilde \Delta$ of error configurations $(\widetilde e^{(1)}, \widetilde e^{(2)}) = \mathcal{T}(e^{(1)}, e^{(2)})$ for $(e^{(1)}, e^{(2)}) \not\in \mathcal C_{\tau}$. 
Every connected component of $\widetilde \Delta$ that is adjacent to $\partial B_C$ is switched, and no switched component is adjacent to $\partial B_A$.  
Then, for errors $(e^{(1)}, e^{(2)}) \not\in \mathcal{C}_\tau$, we have
\begin{equation}
  \sigma(\widetilde e^{(1)})_{\partial B_A}
  =
  \sigma(e^{(1)})_{\partial B_A},
  \qquad
  \sigma(\widetilde e^{(2)})_{\partial B_C}
  =
  \sigma(e^{(1)})_{\partial B_C}.
\end{equation}
Thus, the syndromes sampled from
$\mathsf{P}_\tau(\sigma(e^{(1)})_{\partial B_A},
\sigma(e^{(1)})_{\partial B_C})$ and
$\mathsf{P}_\tau(\sigma(\widetilde e^{(1)})_{\partial B_A})
\mathsf{P}_\tau(\sigma(\widetilde e^{(2)})_{\partial B_C})$ can differ \textit{only} when the pair of sampled error configurations is in $\mathcal{C}_\tau$.
By the coupling lemma (Lemma \ref{lem:coupling}), we then have that
\begin{equation}
  \begin{aligned}
  \delta_\tau
  &:=
  \left\|
  \mathsf{P}_\tau(\Sigma_{\textrm{int},\partial B_A}, \Sigma_{\textrm{int},\partial B_C})
  -
  \mathsf{P}_\tau(\Sigma_{\textrm{int},\partial B_A})
  \mathsf{P}_\tau(\Sigma_{\textrm{int},\partial B_C})
  \right\|_{\rm TV}
  \\
  &\leq
  \Pr_{(e^{(1)},e^{(2)})\sim (q_B^\tau)^{\otimes 2}}
  \left[
  \sigma(e^{(1)})_{\partial B_A \cup \partial B_C}
  \neq
  (\sigma(\widetilde e^{(1)})_{\partial B_A},
  \sigma(\widetilde e^{(2)})_{\partial B_C})
  \right]
  \\
  &\leq
  \Pr_{(e^{(1)},e^{(2)})\sim (q_B^\tau)^{\otimes 2}}
  \left[(e^{(1)},e^{(2)})\in \mathcal{C}_\tau\right].
  \label{eq:tv-crossing-bound}
  \end{aligned}
\end{equation}
\subsection{Weight of the disagreement component}
It remains to bound the probability of sampling a pair of error configurations $(e^{(1)}, e^{(2)}) \in  \mathcal{C}_\tau$.
Whenever $\mathcal{C}_\tau$ occurs, $\supp(\Delta)$ contains a connected
component $\Gamma\subseteq B$ that is adjacent to both $\partial B_A$ and
$\partial B_C$.  With the edge-distance convention used here, such a
component satisfies
\begin{equation}
  |\Gamma|\geq \dist(A,C)_B-1.
\end{equation}

Fix a connected set $\Gamma\subseteq B$ and a nonzero disagreement pattern
$\gamma\in(G\setminus\{0\})^\Gamma$.  We identify $\gamma$ with its
extension by zero to all of $B$, and define
\begin{equation}
  W_\Gamma(x):=\prod_{i\in\Gamma}q_i(x_i),
  \qquad x\in G^\Gamma.
\end{equation}
If $\Gamma$ is a connected component of $\supp(\Delta)$ and
$\gamma=\Delta_\Gamma$, then $\sigma_B(\gamma)=0$.  Indeed, a stabilizer
generator supported in $B$ cannot intersect two distinct components of
$\supp(\Delta)$, since any two qudits in its support are adjacent.  Thus
each component contributes separately to the zero syndrome
$\sigma_B(\Delta)=0$.

For a fixed pair $(\Gamma,\gamma)$ satisfying $\sigma_B(\gamma)=0$, sampling
the common syndrome first is equivalent to sampling $e^{(1)}$ from $q_B$
and then sampling $e^{(2)}$ conditioned on the syndrome of $e^{(1)}$:
\begin{equation}
\begin{split}
  \mathbb E_{\tau\sim\mathsf P_B}
  \Pr_{(e^{(1)},e^{(2)})\sim(q_B^\tau)^{\otimes2}}
  \left(\Delta_\Gamma=\gamma\right)
  =
  \mathbb E_{e^{(1)}\sim q_B}
  \Pr_{e^{(2)}\sim q_B^{\sigma_B(e^{(1)})}}
  \left(e^{(2)}_\Gamma=e^{(1)}_\Gamma+\gamma\right).
\end{split}
\label{eq:component-probability-average}
\end{equation}

Since $\sigma_B(\gamma)=0$, the shift
$e^{(2)}_\Gamma\mapsto e^{(2)}_\Gamma-\gamma$ does not change the syndrome
on $B$.  Consequently, after fixing the outside configuration
$e^{(2)}_{B\setminus\Gamma}$, the choice
$e^{(2)}_\Gamma=e^{(1)}_\Gamma+\gamma$ satisfies the conditioning constraint
$\sigma_B(e^{(2)})=\sigma_B(e^{(1)})$ if and only if the choice
$e^{(2)}_\Gamma=e^{(1)}_\Gamma$ satisfies it.  The two configurations have the same
outside weight and local weights $W_\Gamma(e^{(1)}_\Gamma+\gamma)$ and
$W_\Gamma(e^{(1)}_\Gamma)$, respectively.  Hence
\begin{equation}
\begin{aligned}
    \Pr_{e^{(2)} \sim q_B^{\sigma_B(e^{(1)})}}
\!\left(e^{(2)}_\Gamma=e^{(1)}_\Gamma+\gamma
    \right)
    &\leq
    \frac{
    W_\Gamma(e^{(1)}_\Gamma+\gamma)
    }{
    W_\Gamma(e^{(1)}_\Gamma+\gamma)
    +
    W_\Gamma(e^{(1)}_\Gamma)
    }
    \\
    &\leq
    \sqrt{
    \frac{
    W_\Gamma(e^{(1)}_\Gamma+\gamma)
    }{
    W_\Gamma(e^{(1)}_\Gamma)
    }
    }.
\end{aligned}
\label{eq:local-ratio-bound}
\end{equation}
If $W_\Gamma(e^{(1)}_\Gamma)=0$, this value of $e^{(1)}$ has zero
$q_B$-probability and contributes nothing below.  Thus the ratio is used
only on the support of $q_B$, including for channels with zero-probability
Pauli labels.

Combining Eqs.~\eqref{eq:component-probability-average} and
\eqref{eq:local-ratio-bound} and averaging over $e^{(1)}$ gives
\begin{equation}
\begin{split}
  \mathbb E_{\tau\sim\mathsf P_B}
  \Pr_{(e^{(1)},e^{(2)})\sim(q_B^\tau)^{\otimes2}}
  \left(\Delta_\Gamma=\gamma\right)
  &\leq
  \sum_{x\in G^\Gamma}
  W_\Gamma(x)
  \sqrt{
  \frac{W_\Gamma(x+\gamma)}
       {W_\Gamma(x)}
  }
  \\
  &=
  \sum_{x\in G^\Gamma}
  \sqrt{W_\Gamma(x)W_\Gamma(x+\gamma)}
  \\
  &=
  \prod_{i\in\Gamma}
  \left(
    \sum_{u\in G}
    \sqrt{q_i(u)q_i(u+\gamma_i)}
  \right).
\end{split}
\label{eq:hellinger-component}
\end{equation}
For uniform qudit depolarizing noise, $q_i=q$, and for every nonzero $g\in G$,
\begin{equation}
\begin{split}
  \sum_{u\in G}\sqrt{q(u)q(u+g)} =
  2\sqrt{\frac{p(1-p)}{d^2-1}}
  +
  (d^2-2)\frac{p}{d^2-1}.
\end{split}
\end{equation}
Summing over the $d^2-1$ nonzero values of $g$ gives the depolarizing
expression for $\alpha$ stated in Theorem~\ref{thm:cmi-stability}.
\subsection{Everything together}
Let $E_{\mathrm{cc}}$ be the set of connected, crossing disagreement
patterns on $B$, defined by
\begin{equation}
\begin{split}
  E_{\mathrm{cc}}:=\bigl\{(\Gamma,\gamma):\;&
  \Gamma\subseteq B\text{ is connected},\quad
  \dist(\Gamma,\partial B_A)=\dist(\Gamma,\partial B_C)=1,\\
  &\gamma\in(G\setminus\{0\})^\Gamma,\quad
  \sigma_B(\gamma)=0\bigr\}.
\end{split}
\end{equation}
Every $(\Gamma,\gamma)\in E_{\mathrm{cc}}$ satisfies
$|\Gamma|\geq\dist(A,C)_B-1$.  By applying
Lemma~\ref{lem:lattice_animal}, when $ \alpha=e\frakd\sweight \max_{i\in\Lambda}
\left( \sum_{g\in G\setminus\{0\}}\sum_{u\in G}\sqrt{q_i(u)q_i(u+g)}\right)$, we have:
\begin{equation}
    \begin{split}
        \mathbb{E}_{\tau \sim \mathsf P_B}\delta_{\tau} &\leq \mathbb{E}_{\tau \sim \mathsf P_B}\Pr_{(e^{(1)},e^{(2)})\sim (q_B^\tau)^{\otimes 2}}
        \left[(e^{(1)},e^{(2)})\in \mathcal{C}_\tau\right] \\
        &\leq \sum_{(\Gamma,\gamma)\in E_{\mathrm{cc}}}
        \mathbb E_{\tau\sim\mathsf P_B}
        \Pr_{(e^{(1)},e^{(2)})\sim(q_B^\tau)^{\otimes2}}
        \left(\Delta_\Gamma=\gamma\right) \\
        & \leq \min(\Omega_A,\Omega_C)
        \sum_{m\geq \dist(A, C)_B - 1}\alpha^m
        \\
        &\leq \frac{\min(\Omega_A,\Omega_C)}{1-\alpha}
        \alpha^{\dist(A, C)_B - 1}.
    \end{split}
\end{equation}
Finally, Eq.~\eqref{eq:tv-crossing-bound} converts the crossing estimate into a mutual-information estimate.
A standard continuity bound~\cite{Audenaert_2007} for mutual information gives, for each $\tau$,
\begin{equation}
  I(\Sigma_{\textrm{int},\partial B_A}:\Sigma_{\textrm{int},\partial B_C}|\Sigma_{\textrm{bulk},B} = \tau)
  \leq h_2(\delta_\tau)
  +(\Omega_A+\Omega_C)\delta_\tau\log d\,
\end{equation}
We then average over $\tau$ and use the concavity of $h_2(x)$ alongside Jensen's inequality to see that
\begin{equation}
\begin{split}
    I(\Sigma_{\textrm{int},\partial B_A} : \Sigma_{\textrm{int},\partial B_C} | \Sigma_{\textrm{bulk},B})
    &= \mathbb{E}_{\tau \sim \mathsf P_B}
    \left[ I(\Sigma_{\textrm{int},\partial B_A} : \Sigma_{\textrm{int},\partial B_C}
    | \Sigma_{\textrm{bulk},B} = \tau ) \right] \\
    &\leq \mathbb{E}_{\tau \sim \mathsf P_B}
    \left[ h_2 (\delta_\tau) \right]
    +(\Omega_A+\Omega_C)\log d\,
    \mathbb{E}_{\tau \sim \mathsf P_B} \left[ \delta_\tau \right] \\
    &\leq h_2 \left( \mathbb{E}_{\tau \sim \mathsf P_B} [\delta_\tau] \right)
    +\frakd(|\partial B_A| + |\partial B_C|)\log d\,
    \mathbb{E}_{\tau \sim \mathsf P_B} \left[ \delta_\tau \right].
\end{split}
\end{equation}

For brevity, set
\begin{equation}
  s:=|\partial B_A|+|\partial B_C|,
  \qquad t:=\dist(A,C)_B-1,
  \qquad \overline\delta:=\mathbb E_{\tau\sim\mathsf P_B}[\delta_\tau].
\end{equation}
First assume that
\begin{equation}
  \varepsilon:=\frac{\frakd s}{1-\alpha}\alpha^t\leq\frac12.
\end{equation}
Then $\overline\delta\leq\varepsilon$, and the standard bound
$h_2(x)\leq2(\log2)\sqrt{x(1-x)}$ gives
\begin{equation}
\begin{split}
  h_2(\overline\delta)
  &\leq 2(\log2)\sqrt{\overline\delta}
  \leq 2(\log2)\sqrt{\frac{\frakd s}{1-\alpha}}\,\alpha^{t/2}
  \leq \frakd^2(\log d)s^2\alpha^{t/2}, \\
  \frakd s\log d\,\overline\delta
  &\leq \frac{\frakd^2 s^2\log d}{1-\alpha}\alpha^t
  \leq 2\frakd^2(\log d)s^2\alpha^t.
\end{split}
\end{equation}
Here we used $s\geq2$, $d\geq2$, and $\alpha<0.2$.  Since
$\alpha^t\leq\alpha^{t/2}$, it follows that, with
$c_0=\frac12\log(1/\alpha)$,
\begin{equation}
  I(\Sigma_{\textrm{int},\partial B_A}:\Sigma_{\textrm{int},\partial B_C}|\Sigma_{\textrm{bulk},B})
  \leq3\frakd^2(\log d)s^2e^{-c_0t}.
  \label{eq:final_inequality_CMI}
\end{equation}

It remains to consider $\varepsilon>1/2$.  In this case,
$\alpha^{t/2}>\sqrt{(1-\alpha)/(2\frakd s)}$, and hence, using
$\frakd\geq1$, $s\geq2$, and
$\alpha<0.2$,
\begin{equation}
  3\frakd^2s^2e^{-c_0t}
  =3\frakd^2s^2\alpha^{t/2}
  >3\frakd^2s^2\sqrt{\frac{1-\alpha}{2\frakd s}}
  >2\frakd s.
\end{equation}
The general upper bound
$I(\Sigma_{\textrm{int},\partial B_A}:\Sigma_{\textrm{int},\partial B_C}|\Sigma_{\textrm{bulk},B})
\leq2(\Omega_A+\Omega_C)\log d
\leq2\frakd s\log d$
therefore proves Eq.~\eqref{eq:final_inequality_CMI}
in this case as well.  Combining this estimate with the preceding exact
CMI decomposition proves Eq.~\eqref{eq:main-bound}.


\section{Proof of Theorem 2}
\label{app:lnhn-proof}
In this appendix we prove Theorem~\ref{thm:information_percolation}. At weak noise, Theorem~\ref{thm:cmi-stability} controls the CMI by showing that error clusters are unlikely to cross the buffer. At strong noise, Theorem~\ref{thm:information_percolation} relies on the complementary mechanism: although error clusters may percolate, the noise on $B$ erases the information carried across the buffer, driving $A$ and $C$ toward conditional independence.

\emph{Strategy.}
The proof has three steps.
\begin{enumerate}
    \item\emph{Local classical representation} (Secs.~\ref{sec:lnhn-cleaning}, \ref{sec:lnhn-classical}, and \ref{sec:lnhn-percolation}).
    After removing noise on $A$ and $C$ by data processing, it suffices to consider $\widetilde\rho=\mathcal N_B(\rho_0)$. We introduce a classical record $Y=\mathsf G+E$, where $\mathsf G$ is uniformly distributed over the normalizer $\mathsf N=\mathsf S^\perp$ and is independent of the Pauli error $E$. Adding $\mathsf G$ preserves every stabilizer syndrome while making the record uniform over all Pauli labels consistent with that syndrome. This gives the exact identities
    \begin{align*}
        S(\widetilde\rho_Q)&=H(Y_Q)-|Q|\log d,\\
        I(A:C\mid B)_{\widetilde\rho}&=I(Y_A:Y_C\mid Y_B).
    \end{align*}
    Thus, the quantum CMI is reduced to the classical CMI of a single qudit-indexed random record.

    To make this representation local, we use Lemma~\ref{lem:lnhn-decomposition}. Under the assumption of Theorem~\ref{thm:information_percolation}, the normalizer is generated by the local subgroup $\mathsf H$ together with elements supported on $AD$ or $CD$, where $D:=\Lambda\setminus(A\cup B\cup C)$. 
    Assigning independent uniform coefficients to these generators produces the uniform label $\mathsf G$. 
    The coefficients associated with normalizer generators supported only on the \(A\) or \(C\) side can be removed by data processing, while each buffer observation depends only on the local generator coefficients touching that qudit and its onsite error. This is the local hidden-variable structure used in the information-percolation bound.

    \item\emph{Information percolation} (Sec.~\ref{sec:lnhn-percolation}).        
    Each buffer observation depends on nearby independent generator coefficients through a channel with contraction coefficient at most $\eta_i$. The information-percolation inequality of Polyanskiy and Wu~\cite{PolyanskiyWu2020} bounds the CMI by $\log d$ times the expected number of interface generators connected to the opposite side in an auxiliary percolation process, where each buffer qudit $i$ is independently open with probability $\eta_i$.

    \item\emph{Path counting} (Sec.~\ref{sec:lnhn-path-counting}).
    Any connection between the two interfaces must pass through at least $r_{\mathsf H}=\dist_{\mathsf H}(A,C)_B-1$ distinct buffer qudits. Each generator acts on at most $w_{\mathsf H}$ qudits, and each qudit participates in at most $D_{\mathsf H}$ generators. These constraints limit the number of possible connecting paths, while the probability that any given path remains open is at most $\bar\eta^{\,r_{\mathsf H}}$. A union bound gives a crossing probability of at most $\alpha_{\mathsf H}^{\,r_{\mathsf H}}$, where $\alpha_{\mathsf H}=(w_{\mathsf H}-1)(D_{\mathsf H}-1)\bar\eta$. Summing over the generators at the smaller interface, whose number is bounded by $D_{\mathsf H}\min\{|\partial_{\mathsf H} A|,|\partial_{\mathsf H} C|\}$, yields Eq.~\eqref{eq:lnhn-spatial}.
\end{enumerate}

\subsection{Two-sided normalizer decomposition}
\label{sec:lnhn-cleaning}
 
\emph{Conventions.}
As in the main text, Pauli operators are identified with their labels in $G^\Lambda$, $G=\mathbb Z_d^2$, and two labels commute when the corresponding Paulis do. For a label subgroup $\mathsf F$, its commutant $\mathsf F^\perp$ is the set of labels commuting with every element of $\mathsf F$. Thus, if $\mathsf S$ is the stabilizer label subgroup, then $\mathsf N=\mathsf S^\perp$ is the Pauli normalizer.
For a region $Q\subseteq\Lambda$, we write $x|_Q$ for the restriction of a label $x$ to $Q$ and denote by $\mathsf F(Q)$ the elements of $\mathsf F$ supported entirely in $Q$. In particular, $\mathsf S(Q)$ is the label representation of the stabilizer subgroup $\mathcal S(Q)$.
For a subgroup of labels supported on $Q$, we write $\perp_Q$ for the commutant taken within $G^Q$. Finally, let $D:=\Lambda\setminus(A\cup B\cup C)$, which may be empty.

We will repeatedly use two standard duality properties of Pauli labels. First, the commutation pairing on $G^Q$ is nondegenerate, so for any subgroup $\mathsf F\subseteq G^Q$,
\begin{equation}
    (\mathsf F^{\perp_Q})^{\perp_Q} = \mathsf F,
    \qquad
    |\mathsf F|\,|\mathsf F^{\perp_Q}| = d^{2|Q|}.
    \label{eq:lnhn-duality}
\end{equation}
Second, a label supported on $Q$ commutes with every restriction $f|_Q$, $f\in\mathsf F$, if and only if its zero extension to $\Lambda$ commutes with every element of $\mathsf F$. Hence
\begin{equation}
    (\mathsf F|_Q)^{\perp_Q} = \mathsf F^\perp(Q).
\end{equation}
Applying this identity to $\mathsf S$ and $\mathsf N=\mathsf S^\perp$, and then using Eq.~\eqref{eq:lnhn-duality}, gives
\begin{equation}
    \mathsf S|_Q = \mathsf N(Q)^{\perp_Q},
    \qquad
    \mathsf N|_Q = \mathsf S(Q)^{\perp_Q}.
    \label{eq:lnhn-restriction-duality}
\end{equation}
Thus, the restrictions of stabilizers to $Q$ are precisely the labels that commute with every normalizer element supported in $Q$, while the restrictions of normalizer elements to $Q$ are precisely the labels that commute with every stabilizer supported in $Q$.

\emph{Assumptions.}
We now express the two hypotheses in the algebraic form used in the normalizer decomposition. First, the zero-CMI assumption~\eqref{eq:lnhn-markov} implies
\begin{equation}
    \mathsf S(ABC)=\mathsf S(AB)+\mathsf S(BC),
    \label{eq:lnhn-S-decomposition}
\end{equation}
where the sum denotes componentwise addition of Pauli labels.

Second, the completeness assumption in theorem~\ref{thm:information_percolation} requires $\mathsf H|_B=\mathsf N|_B$. To express this condition in terms of labels supported in $B$, we use Eq.~\eqref{eq:lnhn-restriction-duality}:
\begin{equation}
    \mathsf H|_B = \bigl(\mathsf H^\perp(B)\bigr)^{\perp_B},
    \qquad
    \mathsf N|_B = \mathsf S(B)^{\perp_B}.
\end{equation}
Taking commutants within $B$ and applying Eq.~\eqref{eq:lnhn-duality}, this equality of restrictions is equivalent to
\begin{equation}
    \mathsf H^\perp(B) = \mathsf S(B).
    \label{eq:lnhn-completeness-equiv}
\end{equation}

Together, these two relations yield the following decomposition: every normalizer label is a sum of an element of $\mathsf H$ and labels supported in $AD$ and $CD$. Thus the local generators of $\mathsf H$ suffice to represent all normalizer degrees of freedom acting on the buffer.

\begin{lemma}[Two-sided normalizer decomposition]
\label{lem:lnhn-decomposition}
Under the assumption of theorem~\ref{thm:information_percolation},
\begin{equation}
    \mathsf N = \mathsf H+\mathsf N(AD)+\mathsf N(CD).
    \label{eq:lnhn-N-decomposition}
\end{equation}
\end{lemma}

\begin{proof}
Set $\mathsf T:=\mathsf H+\mathsf N(AD)+\mathsf N(CD)\subseteq\mathsf N$. By Eq.~\eqref{eq:lnhn-duality}, it suffices to show $\mathsf T^\perp\subseteq\mathsf S$. Take $x\in\mathsf T^\perp$. Since $\mathsf S\subseteq\mathsf H$, we have $x\in\mathsf N$. Commutation with $\mathsf N(AD)$ and $\mathsf N(CD)$, together with Eq.~\eqref{eq:lnhn-restriction-duality}, gives stabilizers $s_1,s_2\in\mathsf S$ such that
\begin{equation}
    x=s_1+v=s_2+u,
    \qquad
    v\in\mathsf N(BC),\quad u\in\mathsf N(AB).
    \label{eq:lnhn-two-representatives}
\end{equation}
Since $u-v=s_1-s_2\in\mathsf S(ABC)$, Eq.~\eqref{eq:lnhn-S-decomposition} yields
\begin{equation}
    u-v=s_{AB}+s_{BC},
    \qquad
    s_{AB}\in\mathsf S(AB),\quad s_{BC}\in\mathsf S(BC).
\end{equation}
Consequently,
\begin{equation}
    z:=u-s_{AB}=v+s_{BC}\in\mathsf N(B),
    \qquad x=s_2+s_{AB}+z,
    \label{eq:lnhn-buffer-witness}
\end{equation}
because the two expressions for $z$ are supported in $AB$ and $BC$, respectively. Both $x$ and $s_2+s_{AB}$ commute with $\mathsf H\subseteq\mathsf N$, so $z\in\mathsf H^\perp(B)=\mathsf S(B)$ by Eq.~\eqref{eq:lnhn-completeness-equiv}. Hence $x\in\mathsf S$.
\end{proof}

The two assumptions enter at distinct steps. Eq.~\eqref{eq:lnhn-markov} allows the side-supported representatives to be combined into a buffer-supported label, while the local completeness assumption forces that label to be a stabilizer. Without completeness, a logical degree of freedom in the buffer can sustain a nonzero CMI at strong noise. For example, in the toric code, let $A$ be an annulus winding around the torus and let $B$ consist of its two neighboring annuli. Then $I(A:C\mid B)_{\mathcal N_B(\rho_0)}\to\log 2$ as $p\to\tfrac12$, so the buffer-noise CMI on the right-hand side of Eq.~\eqref{eq:lnhn-outer-DPI} does not vanish in this limit.

\subsection{From the noisy state to a classical record}
\label{sec:lnhn-classical}

Noise acting on $D$ disappears upon tracing out $D$, while noise on $A$ and $C$ can be removed by data processing, exactly as in Eq.~\eqref{eq:data_processing_ineq}. Hence
\begin{equation}
    I(A:C\mid B)_\rho
    \leq
    I(A:C\mid B)_{\widetilde\rho},
    \qquad
    \widetilde\rho:=\mathcal N_B(\rho_0).
    \label{eq:lnhn-outer-DPI}
\end{equation}
It therefore suffices to study noise acting only on the buffer $B$.

\emph{From syndrome entropy to the qudit record.}
Let $E$ be the Pauli error label of $\widetilde\rho$, with $E_i\sim q_i$ independently for $i\in B$ and $E_i=0$ off $B$, and let $\widehat\Sigma_Q:=\chi_E|_{\mathcal S(Q)}$ be its syndrome on $Q$. Each quantum syndrome sector has dimension $d^{|Q|}/|\mathsf S(Q)|$, and $\widetilde\rho_Q$ is maximally mixed within each occupied sector. Thus
\begin{equation}
 S(\widetilde\rho_Q)
 =H(\widehat\Sigma_Q)+|Q|\log d-\log|\mathsf S(Q)|.
 \label{eq:lnhn-quantum-syndrome-entropy}
\end{equation}
Now set $Y=\mathsf G+E$, where $\mathsf G$ is independent of $E$ and uniform over $\mathsf N$. By Eq.~\eqref{eq:lnhn-restriction-duality}, $\mathsf N|_Q=\mathsf S(Q)^{\perp_Q}$ is the kernel of the regional syndrome map. Adding $\mathsf G_Q$ therefore preserves the syndrome and makes $Y_Q$, conditioned on that syndrome, uniform over its $d^{2|Q|}/|\mathsf S(Q)|$ Pauli-label representatives. Since $\widehat\Sigma_Q$ is determined by $Y_Q$, the Shannon entropy chain rule gives
\begin{align}
 H(Y_Q)
 &=H(\widehat\Sigma_Q)+H(Y_Q\mid\widehat\Sigma_Q)\nonumber\\
 &=H(\widehat\Sigma_Q)+2|Q|\log d-\log|\mathsf S(Q)|.
\end{align}
Subtracting the quantum entropy expression yields
\begin{equation}
 S(\widetilde\rho_Q)=H(Y_Q)-|Q|\log d.
 \label{eq:lnhn-exact-entropy}
\end{equation}
Since the records on $AB$, $BC$, $B$, and $ABC$ are restrictions of the same $Y$, the $|Q|\log d$ corrections cancel, giving
\begin{equation}
    I(A:C\mid B)_{\widetilde\rho}=I(Y_A:Y_C\mid Y_B).
    \label{eq:lnhn-exact-CMI}
\end{equation}

\subsection{Local hidden variables and information percolation}
\label{sec:lnhn-percolation}

\emph{Local hidden-variable representation.}
To apply the information-percolation inequality, we realize the uniform normalizer label $\mathsf G$ from Sec.~\ref{sec:lnhn-classical} using independent generator coefficients. By Eq.~\eqref{eq:lnhn-N-decomposition}, the normalizer $\mathsf N$ can be generated by the local labels $\{t_\mu\}_{\mu\in\mathcal J}$ together with labels
\begin{align}
    \ell_j^A\in\mathsf N(AD),
    \qquad
    \ell_k^C\in\mathsf N(CD).
\end{align}
Choose coefficients $h_\mu,u_j,v_k\in\mathbb Z_d$ independently and uniformly, also independently of $E$, and define
\begin{equation}
    \mathsf G :=
    \sum_{\mu\in\mathcal J}h_\mu t_\mu
    +\sum_j u_j\ell_j^A
    +\sum_k v_k\ell_k^C,
    \qquad
    Y:=\mathsf G+E.
    \label{eq:lnhn-hidden-sum}
\end{equation}
Every label in $\mathsf N$ is produced by the same number of coefficient choices, so $\mathsf G$ is uniform over $\mathsf N$. 
$\mathsf G$ is independent of $E$ because the coefficients are chosen independently of the error.
We write $h=(h_\mu)_{\mu\in\mathcal J}$, $u=(u_j)_j$, and $v=(v_k)_k$ for the corresponding coefficient collections.

Let $V_Q:=\{\mu:\supp(t_\mu)\cap Q\neq\varnothing\}$ index the generators touching $Q$, and let $N_{\mathsf H}(i):=\{\mu:i\in\supp(t_\mu)\}$, so $|N_{\mathsf H}(i)|\leq D_{\mathsf H}$. The generator expansion gives
\begin{equation}
\begin{gathered}
    Y_i=F_i(h)+E_i,\qquad
    F_i(h):=\sum_{\mu\in N_{\mathsf H}(i)}h_\mu(t_\mu)_i,
    \quad i\in B,
\end{gathered}
\label{eq:lnhn-outer-functions}
\end{equation}
The records $Y_A$ and $Y_C$ are deterministic functions of $(h_{V_A},u)$ and $(h_{V_C},v)$, respectively. Conditioning on $Y_B$, which depends on $h$ but not on $u,v$, can correlate the $h$ coefficients, while $u,v$ remain mutually independent and independent of $h$. Conditional data processing therefore yields
\begin{equation}
\begin{split}
 I(Y_A:Y_C\mid Y_B)
 &\leq I\bigl((h_{V_A},u):(h_{V_C},v)\mid Y_B\bigr)\\
 &=I(h_{V_A}:h_{V_C}\mid Y_B).
\end{split}
\label{eq:lnhn-hidden-DPI}
\end{equation}
Together with Eqs.~\eqref{eq:lnhn-outer-DPI} and \eqref{eq:lnhn-exact-CMI}, this gives
\begin{equation}
 I(A:C\mid B)_\rho\leq I(h_{V_A}:h_{V_C}\mid Y_B).
 \label{eq:lnhn-ready-for-PW}
\end{equation}
The coefficients $h_\mu$ are independent, and the observations $Y_i$ are conditionally independent given $h$, each depending only on $h_{N_{\mathsf H}(i)}$. This is the structure required by the following information-percolation inequality.

\emph{Information percolation.}
We use the following finite-system form of Theorem~2 of Polyanskiy and Wu~\cite{PolyanskiyWu2020}.
\begin{theorem}[Information-percolation inequality, Polyanskiy--Wu] Let $V,W$ be finite index sets, let $(X_v)_{v\in V}$ be independent discrete random variables with finite entropies, and let $(Z_w)_{w\in W}$ be conditionally independent given $X_V$. For each $w\in W$, let $N(w)\subseteq V$ and suppose the conditional probability distribution of $Z_w$ depends only on $X_{N(w)}$. Assume this channel has input-unrestricted contraction coefficient at most $\theta_w\in[0,1]$: for every joint input probability distribution satisfying $U\to X_{N(w)}\to Z_w$,
\begin{align}
    I(U:Z_w)\leq\theta_w I(U:X_{N(w)}).
\end{align}
Here $U$ is an arbitrary auxiliary classical variable correlated with the channel input.

Associate to each observation vertex $w$ an independent Bernoulli variable that is open with probability $\theta_w$. We say that $v\in V$ is connected to a set $S\subseteq V$ if either $v\in S$, or there exists a sequence of open observations $w_1,\ldots,w_m$ such that
\begin{equation}
    v\in N(w_1),
    \qquad
    N(w_j)\cap N(w_{j+1})\neq\varnothing,
    \qquad
    N(w_m)\cap S\neq\varnothing.
\end{equation}
Then, for any $S_1,S_2\subseteq V$,
\begin{equation}
\begin{split}
    I(X_{S_1}:X_{S_2}\mid Z_W)
    \leq
    \Bigl(\max_{v\in V}H(X_v)\Bigr)
    \sum_{v\in S_1}
    \Pr_{\boldsymbol\theta}(v\leftrightarrow S_2).
\end{split}
\label{eq:lnhn-PW}
\end{equation}
Here $\Pr_{\boldsymbol\theta}(v\leftrightarrow S_2)$ is the probability that $v$ is connected to $S_2$ when each observation vertex $w$ is independently open with probability $\theta_w$.
\end{theorem}

Thus the residual conditional information is controlled by the probability that the two sets of hidden variables remain connected in an auxiliary percolation problem, where observation $w$ is retained with probability $\theta_w$.

We apply Eq.~\eqref{eq:lnhn-PW} with
\begin{equation}
    X_\mu=h_\mu,
    \qquad
    Z_i=Y_i,
    \qquad
    N(i)=N_{\mathsf H}(i),
    \qquad
    \theta_i=\eta_i .
\end{equation}
Since each $h_\mu$ is uniform on $\mathbb Z_d$,
$H(h_\mu)=\log d$. Taking $S_1=V_A$ and $S_2=V_C$ gives
\begin{equation}
    I(h_{V_A}:h_{V_C}\mid Y_B)
    \leq
    (\log d)
    \sum_{\mu\in V_A}
    \Pr_{\boldsymbol\eta}(\mu\leftrightarrow V_C).
    \label{eq:lnhn-PW-applied}
\end{equation}

Let
\begin{equation}
    \Sigma_A:=V_A\cap V_B,
    \qquad
    \Sigma_C:=V_C\cap V_B
\end{equation}
denote the generators crossing the $A$--$B$ and $B$--$C$ interfaces.
Because the only observation vertices lie in $B$, any generator in
$V_A\setminus\Sigma_A$ is disconnected from the buffer and therefore
cannot contribute to a path reaching $V_C$. Likewise, any path from the
$A$ side that reaches $V_C$ must first enter through $\Sigma_A$ and exit
through $\Sigma_C$.

Using Eq.~\eqref{eq:lnhn-ready-for-PW},
Eq.~\eqref{eq:lnhn-PW-applied}, and the symmetry under
$A\leftrightarrow C$, we obtain
\begin{equation}
\begin{split}
    I(A:C\mid B)_\rho
    \leq
    (\log d)\min\Bigl\{&
    \sum_{\mu\in\Sigma_A}
    \Pr_{\boldsymbol\eta}(\mu\leftrightarrow\Sigma_C),
    \sum_{\mu\in\Sigma_C}
    \Pr_{\boldsymbol\eta}(\mu\leftrightarrow\Sigma_A)
    \Bigr\}.
\end{split}
\label{eq:lnhn-percolation-bound}
\end{equation}

The probabilities in Eq.~\eqref{eq:lnhn-percolation-bound} refer to the
auxiliary Bernoulli percolation process in which each buffer qudit $i$
is independently open with probability $\eta_i$. In
particular, if $\eta_i=0$ for every $i\in B$, as for completely
randomizing noise on the buffer, then the CMI vanishes exactly.

\emph{Depolarizing noise.}
For $q_i(0)=1-p_i$ and $q_i(u)=p_i/(d^2-1)$ for $u\neq0$, set $a_i:=1-p_i$ and $b_i:=p_i/(d^2-1)$. By Theorem~1 of Ref.~\cite{OrdentlichPolyanskiy2020Binary}, the input-unrestricted contraction coefficient is the maximum over subchannels restricted to two input labels. The output alphabet remains $G$, and all input pairs are equivalent by symmetry. Fix inputs $x\neq x'$. The outputs $x,x'$ have probabilities $(a_i,b_i)$ or $(b_i,a_i)$, respectively, while every other output has probability $b_i$ under either input.

Let $J$ indicate the event $Y\in\{x,x'\}$. Since $\Pr(J=1\mid X)=a_i+b_i$ for either input, $J$ is independent of $(U,X)$ for every $U\to X\to Y$. Conditional on $J=0$, the output distribution is identical for both inputs, so this branch carries no information. Because $J$ is determined by $Y$, the chain rule gives
\begin{align*}
 I(U:Y)&=I(U:J)+I(U:Y\mid J)\\
       &=(a_i+b_i)I(U:Y\mid J=1).
\end{align*}
Thus conditioning on $J=1$ retains the probability of this branch rather than discarding the other outputs. On this branch, the channel is binary symmetric with crossover probability $\delta_i=b_i/(a_i+b_i)$, and the joint input distribution of $(U,X)$ is unchanged. Equation~(16) of Ref.~\cite{PolyanskiyWu2017SDPI} gives its contraction coefficient as $(1-2\delta_i)^2$. Taking the supremum of the information ratio therefore yields
\begin{equation}
    \eta_i=(a_i+b_i)\left(1-\frac{2b_i}{a_i+b_i}\right)^2
    =\frac{\bigl(1-p_i-p_i/(d^2-1)\bigr)^2}
    {1-p_i+p_i/(d^2-1)}.
    \label{eq:lnhn-depolarizing-eta}
\end{equation}
 
\subsection{Path counting}
\label{sec:lnhn-path-counting}

The connection probabilities in Eq.~\eqref{eq:lnhn-percolation-bound}
are monotone in the opening probabilities. We may therefore replace all
$\eta_i$ by their maximum $\bar\eta$. Define
\begin{equation}
\begin{gathered}
 \bar\eta:=\max_{i\in B}\eta_i,
 \qquad
 \alpha_{\mathsf H}:=(w_{\mathsf H}-1)(D_{\mathsf H}-1)\,\bar\eta, \qquad
 r_{\mathsf H}:=\dist_{\mathsf H}(A,C)_B-1\geq1.
\end{gathered}
 \label{eq:lnhn-alpha}
\end{equation}

Fix $\mu\in\Sigma_A$ and suppose that it is connected to $\Sigma_C$. Then there exists a simple alternating path
\begin{equation}
    \mu=\mu_0,i_1,\mu_1,\ldots,i_m,\mu_m,
\end{equation}
where the $i_k$ are distinct open qudits in $B$,
\begin{equation}
    i_k\in
    \supp(t_{\mu_{k-1}})
    \cap
    \supp(t_{\mu_k}),
\end{equation}
and $\mu_m\in\Sigma_C$. Choose
$a\in A\cap\supp(t_{\mu_0})$ and
$c\in C\cap\supp(t_{\mu_m})$. Then
\begin{equation}
    a,i_1,\ldots,i_m,c
\end{equation}
forms a path from $A$ to $C$ in $\Gamma_{\mathsf H}$ whose internal vertices lie in $B$. Hence
\begin{equation}
    m\geq\dist_{\mathsf H}(A,C)_B-1=r_{\mathsf H}.
    \label{eq:lnhn-incidence-distance}
\end{equation}
Thus every crossing contains at least $r_{\mathsf H}$ open buffer qudits. It is therefore enough to count the prefix ending at the $r_{\mathsf H}$-th such qudit.

Recall that every generator $t_\mu$ has weight at most $w_{\mathsf H}$, while each qudit belongs to the support of at most $D_{\mathsf H}$ generators. Since $\mu_0$ already intersects $A$, there are at most $w_{\mathsf H}-1$ choices for the first buffer qudit. At each subsequent buffer qudit, there are at most $D_{\mathsf H}-1$ choices for the next generator, excluding the generator just used, and at each new generator there are at most $w_{\mathsf H}-1$ choices for the next qudit. The number of possible prefixes is therefore bounded by
\begin{equation}
    (w_{\mathsf H}-1)
    \bigl[(D_{\mathsf H}-1)(w_{\mathsf H}-1)\bigr]^{r_{\mathsf H}-1}
    \leq
    \bigl[(w_{\mathsf H}-1)(D_{\mathsf H}-1)\bigr]^{r_{\mathsf H}}.
    \label{eq:lnhn-prefix-count}
\end{equation}

Each such prefix requires $r_{\mathsf H}$ distinct buffer qudits to be open, which under the enlarged opening probability occurs with probability $\bar\eta^{\,r_{\mathsf H}}$. A union bound therefore gives
\begin{equation}
\begin{split}
    \Pr_{\boldsymbol\eta}
    (\mu\leftrightarrow\Sigma_C)
    &\leq
    \Pr_{\bar\eta}
    (\mu\leftrightarrow\Sigma_C)
    \\
    &\leq
    \bigl[(w_{\mathsf H}-1)(D_{\mathsf H}-1)\bar\eta\bigr]^{r_{\mathsf H}}
    =
    \alpha_{\mathsf H}^{\,r_{\mathsf H}},
\end{split}
\label{eq:lnhn-path-probability}
\end{equation}
where $\alpha_{\mathsf H}=(w_{\mathsf H}-1)(D_{\mathsf H}-1)\bar\eta$. The same bound holds with $A$ and $C$ exchanged.

It remains to count the generators at the two interfaces. Every $\mu\in\Sigma_A$ contains at least one qudit in $A$ and one in $B$. The $A$-side qudit therefore belongs to $\partial_{\mathsf H} A$. Since each qudit lies in at most $D_{\mathsf H}$ generator supports,
\begin{equation}
    |\Sigma_A|
    \leq
    D_{\mathsf H}|\partial_{\mathsf H} A|,
    \qquad
    |\Sigma_C|
    \leq
    D_{\mathsf H}|\partial_{\mathsf H} C|.
    \label{eq:lnhn-boundary-count}
\end{equation}
Substituting Eqs.~\eqref{eq:lnhn-path-probability} and
\eqref{eq:lnhn-boundary-count} into Eq.~\eqref{eq:lnhn-percolation-bound}, and using
\begin{equation}
    \alpha_{\mathsf H}^{\,r_{\mathsf H}}
    =
    e^{-c_{\mathsf H}[\dist_{\mathsf H}(A,C)_B-1]},
    \qquad
    c_{\mathsf H}:=\log(1/\alpha_{\mathsf H}),
\end{equation}
proves Eq.~\eqref{eq:lnhn-spatial}.
\hfill$\square$

\section{The Toric Code under $X$-dephasing}
\label{sec:lnhn-toric}
 
For toric code under $X$-dephasing, the classical record splits
into independent $X$ and $Z$ sectors. The reference Markov condition
makes the noiseless $Z$-sector CMI vanish, so only the noisy $X$ sector
contributes. Qubits lie on the edges of an $L\times L$ periodic square
lattice, with stars $A_v=\prod_{e\ni v}X_e$ and plaquettes
$B_p=\prod_{e\in\partial p}Z_e$. The local $X$-sector generators are
the stars: each has weight four and each qubit belongs to two stars,
so $w_X=4$ and $D_X=2$.
The onsite noise is
\begin{equation}
 q(0)=1-p,\qquad q(x)=p,\qquad q(y)=q(z)=0,
 \label{eq:toric-noise}
\end{equation}
where $x=(1,0)\in G$. Only plaquettes incident to an odd number of
error qubits carry syndrome. Throughout, $A\sqcup B\sqcup C=\Lambda$
with $B$ separating $A$ from $C$, as in
Fig.~\ref{fig:cmi-partitions}(a).

The $X$-sector observation channel is binary symmetric with crossover
probability $p$, whose contraction coefficient is~\cite[Eq.~(16)]{PolyanskiyWu2017SDPI}
\begin{equation}
    \eta^{(X)}=(1-2p)^2.
    \label{eq:toric-sector-eta}
\end{equation}

\begin{appendixcorollary}[Toric-code noise thresholds]
\label{cor:toric-thresholds}
Let $\rho=\noise(\rhozero)$ for the toric code under independent $X$-dephasing with probability $0\leq p\leq\tfrac12$, and let $A\sqcup B\sqcup C=\Lambda$ with $B$ separating $A$ from $C$.
\begin{enumerate}
\item[(i)] The error-cluster argument gives exponential CMI decay whenever
\begin{equation}
 6\sqrt{p(1-p)}<1,
 \qquad\text{equivalently}\qquad
 p<p_{\mathrm{EC}}:=\frac12-\frac{\sqrt2}{3}\simeq0.029.
 \label{eq:toric-EC}
\end{equation}
\item[(ii)] If $B$ contains no non-contractible loop of the lattice or its dual, the information-percolation argument gives exponential CMI decay whenever
\begin{equation}
 (1-2p)^2<\frac12,
 \qquad\text{equivalently}\qquad
 p>p_{\mathrm{IP}}:=\frac12-\frac{\sqrt2}{4}\simeq0.146.
 \label{eq:toric-IP}
\end{equation}
\end{enumerate}
More precisely, for each fixed $p$ in either regime there are constants $K(p),c(p)>0$, independent of system size, such that
\begin{equation}
 \begin{split}
 I(A:C\mid B)_\rho\leq{}&K(p)\bigl(|\partial B_A|+|\partial B_C|\bigr)^2 e^{-c(p)[\dist(A,C)_B-1]}.
 \end{split}
\end{equation}
\end{appendixcorollary}

The weak-noise threshold comes from convergence of a weighted path sum; the strong-noise threshold comes from the exact critical density $1/2$ of square-lattice bond percolation. These are sufficient thresholds for the two arguments, not claims about the exact transition. The Nishimori point $p_c\simeq0.109$ of the random-bond Ising model~\cite{Dennis_2002} lies between them. Using the square-lattice geometry improves both bounds: the general estimate in Theorem~\ref{thm:cmi-stability} would only give $p<5\times10^{-6}$, while the generic path count in Sec.~\ref{sec:lnhn-path-counting} would require $3(1-2p)^2<1$, or $p>0.211$.

\emph{Proof of (i).}
Only plaquettes carry syndrome under $X$-dephasing, so the coupling argument of Appendix~\ref{sec:main-proof} uses connected dual-edge components. Every crossing component contains a simple dual path $\pi$ of length at least $\dist(A,C)_B-1$. The $X$ pattern on this path has zero bulk syndrome, since its interior vertices have even degree and its endpoints lie on interface plaquettes. The averaged Hellinger bound~\eqref{eq:hellinger-component} therefore assigns it weight at most $(2\sqrt{p(1-p)})^{|\pi|}$. There are at most $2|\partial B_A|$ starting plaquettes and $4\cdot3^{m-1}$ paths of length $m$ from each, giving
\begin{equation}
 \mathbb E_\tau\delta_\tau
 \leq\frac83|\partial B_A|
 \sum_{m\geq\dist(A,C)_B-1}\bigl(6\sqrt{p(1-p)}\bigr)^m.
 \label{eq:toric-crossing}
\end{equation}
This geometric tail decays exponentially when $6\sqrt{p(1-p)}<1$, or $p<p_{\mathrm{EC}}$. The entropy-continuity argument of Appendix~\ref{sec:main-proof}, with $\Omega_A+\Omega_C\leq2(|\partial B_A|+|\partial B_C|)$, then gives the claimed CMI bound. \hfill$\square$

\emph{Proof of (ii).}
The absence of non-contractible loops in $B$ implies $\mathsf N(B)=\mathsf S(B)$, so completeness holds with $\mathsf H=\mathsf S$. For the $X$ sector identified above, removing the one-sided logical coefficients as in Eq.~\eqref{eq:lnhn-hidden-DPI} leaves the buffer observations
\begin{equation}
 Y_e=h_v+h_{v'}+E_e\pmod2,
 \qquad e=(v,v')\in B,
 \label{eq:toric-record}
\end{equation}
where the star coefficients $h_v$ are independent uniform bits. By Eq.~\eqref{eq:toric-sector-eta}, Eq.~\eqref{eq:lnhn-percolation-bound} reduces the CMI bound to square-lattice bond percolation with open-edge probability $(1-2p)^2$. The no-loop assumption allows each buffer component to lift to the planar square lattice, where connection probabilities decay exponentially below the critical density $1/2$~\cite{Kesten1980,Menshikov1986,AizenmanBarsky1987,DuminilCopinTassion2016}. Thus Eq.~\eqref{eq:toric-IP} gives exponential decay. Crossings require at least $\dist(A,C)_B-1$ edges, and there are at most $2(|\partial B_A|+|\partial B_C|)$ interface vertices, yielding the claimed bound. \hfill$\square$

\section{Geometrically $k$-local Pauli noise}
\label{sec:proof-k-local-corollary}

\begin{corollary}[Geometrically $k$-local Pauli noise]
\label{coro:k-local-pauli-noise}
Let $\rho_0$ be the normalized code-projector state of a $\mathbb Z_d$
qLDPC code satisfying the assumptions of
Theorem~\ref{thm:cmi-stability}.  Let
$\cL=\{\mu_1,\ldots,\mu_M\}$ be a collection of qudit sets, where each
$\mu\in\cL$ supports an elementary local Pauli noise channel
$\mathcal N_\mu$.  The channel is characterized by an error distribution
$q_\mu$ over Pauli operators supported in $\mu$.  Write a Pauli label as a
function
\begin{align}
    u^\mu: \mu \longrightarrow \mathbb Z_d^2,
    \qquad
    u^\mu(i)=:u_i^\mu=(u_{i,x}^\mu,u_{i,z}^\mu),
\end{align}
where $u_{i,x}^\mu,u_{i,z}^\mu\in\mathbb Z_d$, and define
\begin{align}
    P(u^\mu)
    :=
    \prod_{i\in\mu}X_i^{u_{i,x}^\mu}Z_i^{u_{i,z}^\mu}.
\end{align}
The elementary and total noise channels are
\begin{align}
    \mathcal N_\mu(\rho)
    &=
    \sum_{u^\mu}q_\mu(u^\mu)
    P(u^\mu)\rho P^\dagger(u^\mu),
    \label{eq:k-local-elementary-channel}\\
    \mathcal N
    &=
    \mathcal N_{\mu_M}\circ\mathcal N_{\mu_{M-1}}
    \circ\cdots\circ\mathcal N_{\mu_1}.
\end{align}
Assume the noise is geometrically $k$-local: every elementary channel acts
on at most $k$ qudits contained within graph diameter at most $k$, and every
qudit belongs to the support of at most $\noiseweight$ elementary channels,
so that
\begin{equation}
    |\mu|,\operatorname{diam}(\mu)\leq k,
    \qquad
    \max_{i\in\Lambda}
    \left|\{\mu\in\cL:i\in\mu\}\right|
    \leq\noiseweight.
    \label{eq:geometric-k-local-assumption}
\end{equation}
Define
\begin{equation}
  \alpha'
  :=\mathrm e k\noiseweight\frakd\sweight
  \max_{\mu\in\cL}
    \sum_{g^\mu\neq0}\sum_{u^\mu}
    \sqrt{q_\mu(u^\mu)q_\mu(u^\mu+g^\mu)},
  \label{eq:k-local-alpha}
\end{equation}
where $g^\mu$ ranges over nonzero functions and $u^\mu$ over all functions
from $\mu$ to $\mathbb Z_d^2$.  If $\alpha'<0.2$, then for every
tripartition $A\sqcup B\sqcup C=\Lambda$ such that $B$ separates $A$ from
$C$, there exist system-size-independent constants $c_0,K>0$ such that
\begin{equation}
    I(A:C|B)_{\mathcal N(\rho_0)}
    \leq
    K e^{c_0}
    \bigl(|\partial B_A|+|\partial B_C|\bigr)^2
    e^{-c_0\dist(A,C)_B},
    \label{eq:k-local-cmi-bound}
\end{equation}
where one may take
$c_0=\frac{1}{2(k+1)}\log(1/\alpha')$.
\end{corollary}

\begin{proof}
The setup and the reduction to classical syndrome variables are the same as in the main text and Appendix~\ref{sec:main-proof}.  We first construct the corresponding classical error model for the elementary noise channels.  If $P$ and $Q$ are Pauli operators, their conjugation channels commute because $PQ$ and $QP$ differ only by a phase and the phase is cancelled for conjugation. Thus the elementary channels may be reordered. Channels supported entirely in $A$ or entirely in $C$ can then be removed by data processing. Channels crossing the $A$--$B$ or $B$--$C$ interface are retained.  Their supports may extend a distance at most $k$ from the corresponding interface, and this finite range is accounted for explicitly below when the length of an event path is converted into qudit-graph distance.

As in the onsite-noise proof, it suffices to retain the syndromes of the chosen local stabilizer generators.  The complete regional syndrome may also contain nonlocal stabilizers generated by products of local checks, but Eq.~\eqref{eq:regional-syndrome-cmi-difference} upper-bounds the physical CMI by the CMI of the local boundary and bulk check syndromes.  Using these local checks is also what ensures that the elementary-event graph has bounded degree.

For each retained elementary channel $\mu\in\cL$, let
\begin{equation}
    U_\mu\sim q_\mu, \qquad U := (U_\mu)_{\mu\in\mathcal{L}}
\end{equation}
denote its sampled Pauli label.  The variables $U_\mu$ are independent, with
joint distribution
\begin{equation}
    Q(U)=\prod_{\mu\in\cL}q_\mu(U_\mu).
    \label{eq:event-product-distribution}
\end{equation}
Their aggregate qudit error is the sitewise sum
\begin{equation}
    E(U)_i
    :=
    \sum_{\mu\ni i}U_\mu(i)
    \quad\text{in }\mathbb Z_d^2.
    \label{eq:aggregate-event-error}
\end{equation}
Up to an irrelevant phase,
$P(E(U))=\prod_{\mu\in\cL}P(U_\mu)$, so averaging over $Q$ reproduces the
composed noise channel.  In particular, every stabilizer syndrome is an
additive function of the elementary labels.  Let $\sigma_B(U)$ denote the
syndrome of $E(U)$ for the local stabilizer generators supported entirely in
$B$, and define
\begin{equation}
    Q_B^\tau(U)
    :=
    \Pr_Q(U| \sigma_B(U)=\tau).
\end{equation}

We next construct the elementary-event graph.  Its vertices are the retained
labels $\mu\in\cL$, and distinct vertices $\mu$ and $\nu$ are adjacent if
some local stabilizer generator supported in $B$ intersects both supports.
Fix a bulk syndrome $\tau$ and draw independently
$U^{(1)},U^{(2)}\sim Q_B^\tau$.  Their eventwise disagreement is
\begin{equation}
    G_\mu:=U_\mu^{(2)}-U_\mu^{(1)}.
\end{equation}
Since the two samples have the same bulk syndrome, additivity gives
$\sigma_B(G)=0$.

Crucially, every connected component $\Gamma$ of the nonzero disagreement
events is separately syndrome-neutral in $B$.  Indeed, fix a local generator
$S_a$ supported in $B$.  All nonzero events whose supports intersect
$\supp(S_a)$ form a clique in the event graph and hence belong to the same
connected component.  Therefore, if $S_a$ intersects $\Gamma$, the syndrome
of $G|_\Gamma$ on $S_a$ equals the full syndrome of $G$ on $S_a$, which is
zero.  If it does not intersect $\Gamma$, the contribution is trivially
zero.  Thus
\begin{equation}
    \sigma_B(G|_\Gamma)=0
    \qquad\text{for every disagreement component }\Gamma.
    \label{eq:event-component-neutrality}
\end{equation}

Let $\mathcal V_A$ and $\mathcal V_C$ be the elementary events that can
affect the $A$- and $C$-side boundary syndrome registers, respectively.
Let $\mathcal R$ be the union of all disagreement components that intersect
$\mathcal V_C$ but not $\mathcal V_A$, and define the switching map
\begin{equation}
\begin{aligned}
    \widetilde U_\mu^{(1)}&=U_\mu^{(2)},&
    \widetilde U_\mu^{(2)}&=U_\mu^{(1)},
    &&\mu\in\mathcal R,\\
    \widetilde U_\mu^{(1)}&=U_\mu^{(1)},&
    \widetilde U_\mu^{(2)}&=U_\mu^{(2)},
    &&\mu\notin\mathcal R.
\end{aligned}
\label{eq:elementary-event-switching}
\end{equation}
Equation~\eqref{eq:event-component-neutrality} shows that transferring the
disagreement on any component between the two samples leaves both bulk
syndromes equal to $\tau$.  Hence the switched configurations remain in the
conditioned sample space.  The map is an involution.  Moreover, for each
$\mu\in\mathcal R$ it only exchanges the pair
$(U_\mu^{(1)},U_\mu^{(2)})$.  The product form of $Q$ therefore gives
\begin{equation}
\begin{split}
    Q_B^\tau(U^{(1)})Q_B^\tau(U^{(2)})
    &=
    Q_B^\tau(\widetilde U^{(1)})
    Q_B^\tau(\widetilde U^{(2)}).
\end{split}
\label{eq:event-switching-weight}
\end{equation}

For completeness, the coupled variables can now be written explicitly as
\begin{equation}
\begin{aligned}
    X&:=\bigl(
      \Sigma_{\textrm{int},\partial B_A}(\widetilde U^{(1)}),
      \Sigma_{\textrm{int},\partial B_C}(\widetilde U^{(1)})
    \bigr),\\
    Y&:=\bigl(
      \Sigma_{\textrm{int},\partial B_A}(U^{(1)}),
      \Sigma_{\textrm{int},\partial B_C}(U^{(2)})
    \bigr).
\end{aligned}
\label{eq:k-local-coupled-variables}
\end{equation}
The invariance in Eq.~\eqref{eq:event-switching-weight} implies that $X$ has
the conditional joint boundary-syndrome distribution.  Independence of the original samples implies that $Y$ has the product of its conditional
marginals.  If no connected disagreement component intersects both $\mathcal V_A$ and
$\mathcal V_C$, the switching leaves the $A$-side syndrome of the first
sample unchanged and replaces its $C$-side syndrome by that of the second
sample.  Hence $X=Y$ outside the crossing event.  The coupling lemma then gives
\begin{equation}
\begin{split}
    \delta_\tau
    &\leq
    \Pr_{(Q_B^\tau)^{\otimes2}}
    \left[
        \begin{array}{c}
        \text{some connected elementary-event disagreement component}\\
        \text{intersects both }\mathcal V_A\text{ and }\mathcal V_C
        \end{array}
    \right].
\end{split}
\label{eq:k-local-tv-crossing}
\end{equation}

It remains to bound the crossing probability.  For a connected event set
$\Gamma$ and a nonzero disagreement pattern
$g=(g^\mu)_{\mu\in\Gamma}$ satisfying $\sigma_B(g)=0$, the same
two-configuration comparison used in Eq.~\eqref{eq:hellinger-component}
gives
\begin{equation}
\begin{split}
    \mathbb E_\tau
    \Pr_{(Q_B^\tau)^{\otimes2}}
    \left[
        G_\mu=g^\mu
        \text{ for every }\mu\in\Gamma
    \right]\leq
    \prod_{\mu\in\Gamma}
    \left(
        \sum_{u^\mu}
        \sqrt{
            q_\mu(u^\mu)
            q_\mu(u^\mu+g^\mu)
        }
    \right).
\end{split}
\label{eq:k-local-hellinger-component}
\end{equation}
Thus the onsite factorization over qudits is replaced by a factorization over
independent elementary events.

We now justify the degree and polymer-counting bounds.  Fix an event $\mu$.
Its at most $k$ qudits meet at most $k\frakd$ local stabilizer generators.
The union of these generator supports contains at most
$k\frakd\sweight$ qudits, and each such qudit belongs to at most
$\noiseweight$ elementary supports.  Consequently, the closed neighborhood
of $\mu$ in the event graph, including $\mu$ itself, has size at most
$k\noiseweight\frakd\sweight$, and hence
\begin{equation}
    \Delta_{\mathrm{ev}}+1
    \leq
    k\noiseweight\frakd\sweight.
    \label{eq:event-graph-degree}
\end{equation}
Set
\begin{equation}
    \eta
    :=
    \max_{\mu\in\cL}
    \sum_{g^\mu\neq0}\sum_{u^\mu}
    \sqrt{q_\mu(u^\mu)q_\mu(u^\mu+g^\mu)},
\end{equation}
so that
$\alpha'=\mathrm e k\noiseweight\frakd\sweight\eta$.
For a fixed event $\mu_0$, Lemma~\ref{lem:lattice_animal},
Eqs.~\eqref{eq:k-local-hellinger-component} and
\eqref{eq:event-graph-degree}, and summation over all nonzero labels on a
connected set give
\begin{equation}
\begin{split}
 \sum_{\substack{\Gamma\ni\mu_0\\|\Gamma|=m}}
   \sum_{\substack{g:\,\supp(g)=\Gamma\\\sigma_B(g)=0}}
   \mathbb E_\tau
   \Pr_{(Q_B^\tau)^{\otimes2}}(G|_\Gamma=g)\leq
 \bigl[\mathrm e(\Delta_{\mathrm{ev}}+1)\eta\bigr]^m
 \leq(\alpha')^m.
\end{split}
\label{eq:k-local-polymer-sum}
\end{equation}

We finally relate event-component size to the spatial separation.  For two
adjacent events, choose one qudit from each that belongs to their common
stabilizer generator; these two qudits have interaction-graph distance one.
Traversing the $m$ event supports costs at most $mk$, the $m-1$ local checks
linking consecutive events contribute $m-1$ edges, and the two checks at the
$A$--$B$ and $B$--$C$ interfaces contribute two further edges.  Thus an
$m$-event crossing produces an $A$-to-$C$ qudit path of length at most
$mk+(m-1)+2=m(k+1)+1$.  Consequently, a
component joining $\mathcal V_A$ to $\mathcal V_C$ must contain at least
\begin{equation}
    m_0
    :=
    \max\left\{1,
    \left\lceil
        \frac{\dist(A,C)_B-1}{k+1}
    \right\rceil
    \right\}
\end{equation}
elementary events.
Bounded incidence and bounded range also give constants $c_A,c_C$ independent of the system size, such that $|\mathcal V_A|\leq c_A|\partial B_A|$ and $|\mathcal V_C|\leq c_C|\partial B_C|$.  Summing
Eq.~\eqref{eq:k-local-polymer-sum} over a starting event on the smaller side and then over $m\geq m_0$ yields
\begin{equation}
\begin{split}
    \mathbb E_\tau\delta_\tau
    &\leq
    C_k\min\{|\partial B_A|,|\partial B_C|\}
    \sum_{m\geq m_0}(\alpha')^m\\
    &\leq
    \frac{C_k}{1-\alpha'}
    \min\{|\partial B_A|,|\partial B_C|\}
    (\alpha')^{m_0},
\end{split}
\label{eq:k-local-averaged-tv}
\end{equation}
where $C_k$ depends only on the bounded local parameters.  Since
\begin{equation}
    c_0
    =
    \frac{1}{2(k+1)}
    \log\frac{1}{\alpha'},
\end{equation}
the last factor in Eq.~\eqref{eq:k-local-averaged-tv} satisfies
\begin{equation}
    (\alpha')^{m_0}
    \leq
    \exp\!\left[-\frac{\log(1/\alpha')}{k+1}
    \bigl(\dist(A,C)_B-1\bigr)\right]
    =e^{2c_0}e^{-2c_0\dist(A,C)_B}.
\end{equation}
Applying the same entropy-continuity bound and the same small-/large-error
case split used at the end of Appendix~\ref{sec:main-proof} takes the square
root of this decay rate and contributes at most one additional boundary-size
factor.  The square root of the prefactor $e^{2c_0}$ gives the explicit
$e^{c_0}$ in Eq.~\eqref{eq:k-local-cmi-bound}; the remaining bounded local
constants are absorbed into $K$.
\end{proof}

\section{Full-region dephasing in the counterexample}
\label{sec:full-dephasing-counterexample}

Apply independent $Z$ dephasing with probability $p$ to every qubit in
$A_LB_LC_L$.  For each qubit let $\xi\in\{0,1\}$ be its Bernoulli error
variable, with $\Pr(\xi=1)=p$, and define its parity bias
\begin{equation}
    \lambda:=\mathbb E[(-1)^\xi]=1-2p.
\end{equation}
Write $A=A_L$, $B=B_L$, $C=C_L$, and $\rho=\mathcal N_{ABC}(\rho_0)$.
Consider the two four-qubit $X$ stabilizers supported on
$\{(0,0),(L,0),(0,H),(L,H)\}\subset AB$ and
$\{(L,0),(L+1,0),(L,1),(L+1,1)\}\subset BC$.
Their syndrome bits are
\begin{align}
 U&=\xi_{L,0}\oplus\xi_{0,0}\oplus\xi_{0,H}\oplus\xi_{L,H},\\
 V&=\xi_{L,0}\oplus\xi_{L+1,0}\oplus\xi_{L,1}\oplus\xi_{L+1,1}.
\end{align}
Let $\widehat\Sigma_Q$ denote the complete syndrome on $Q$.
Since the $B$ syndrome is trivial, the stabilizer entropy formula gives
\begin{align}
 I(A:C\mid B)_\rho
 ={}&I(A:C\mid B)_{\rho_0}
 +I(\widehat\Sigma_{AB}:\widehat\Sigma_{BC})\nonumber\\
 &-H(\widehat\Sigma_{ABC}\mid\widehat\Sigma_{AB},\widehat\Sigma_{BC})
 \geq I(U:V).
 \label{eq:full-dephasing-syndrome-lower-bound}
\end{align}
Here the conditional entropy is at most
$\log[|\mathcal S(ABC)|/|\mathcal S(AB)\mathcal S(BC)|]
=I(A:C\mid B)_{\rho_0}$, and retaining only $U,V$ cannot increase
mutual information. The two parities share only $\xi_{L,0}$, so
\begin{align}
    \mathbb E[(-1)^{U+V}]&=\lambda^6,\\
    \mathbb E[(-1)^U]\,\mathbb E[(-1)^V]&=\lambda^8.
\end{align}
The joint distribution thus differs from the product of its marginals by
\begin{equation}
    \Delta(p)
    :=\lambda^6-\lambda^8
    =4p(1-p)(1-2p)^6>0,
    \qquad 0<p<\frac12.
    \label{eq:full-dephasing-character-gap}
\end{equation}
To relate $\Delta$ to the total variation, let
$P:=P_{UV}$ and $Q:=P_U P_V$.  Since
$f(u,v):=(-1)^{u+v}$ satisfies $|f(u,v)|=1$,
\begin{align}
    \Delta(p)
    &=
    \left|
        \sum_{u,v}f(u,v)\bigl[P(u,v)-Q(u,v)\bigr]
    \right| \nonumber\\
    &\leq
    \sum_{u,v}
    \left|P_{UV}(u,v)-P_U(u)P_V(v)\right|
    \nonumber\\
    &=2d_{\rm TV}\!\left(P_{UV},P_UP_V\right).
    \label{eq:full-dephasing-tv-bound}
\end{align}
Thus $d_{\rm TV}(P_{UV},P_UP_V)\geq\Delta(p)/2$.  Pinsker's
inequality in nats, $I(U:V)\geq2d_{\rm TV}(P_{UV},P_UP_V)^2$, gives
\begin{equation}
    I(A:C\mid B)_\rho\geq I(U:V)
    \geq\frac{\Delta(p)^2}{2}
    =8p^2(1-p)^2(1-2p)^{12}>0.
    \label{eq:full-dephasing-mi-lower-bound}
\end{equation}
This lower bound is independent of $\dist(A,C)_B$. 

\section{Local stabilizer generation from space-uniformity}
\label{sec:space-uniformity-local-generation}
We now detail the recursive construction used in the proof of Corollary~\ref{cor:space-uniformity-stability}. Decomposing a region into local patches that are conditionally independent is a standard approach in settings with vanishing or small conditional mutual information (CMI)~\cite{yang2024classicalneuralnetworksrepresent,lee2025classicalsimulationnoisyrandom,kumar2026unlearnablephasesmatter}.
\begin{equation}
    B=Q_1\cup\cdots\cup Q_N
\end{equation}
by local, simply connected patches such that every partial union is simply
connected.  Set $R_1:=Q_1$ and, for $j=2,\ldots,N$, define recursively
\begin{equation}
    R_j:=R_{j-1}\cup Q_j,
    \qquad
    M_j:=R_{j-1}\cap Q_j,
\end{equation}
together with
\begin{equation}
    X_j:=R_{j-1}\setminus M_j,
    \qquad
    Z_j:=Q_j\setminus M_j.
\end{equation}
The patches may be chosen and ordered so that $M_j$ separates $X_j$ from
$Z_j$ and the CMI of tripartition $(X_j,M_j,Z_j)$ is obtained from
Fig.~\ref{fig:cmi-partitions}(b) by an allowed deformation or data processing inequality.  The
space-uniformity assumption therefore gives
\begin{equation}
    I(X_j:Z_j| M_j)_{\rho_0}=0.
\end{equation}
Because
\begin{equation}
    R_j=X_jM_jZ_j,
    \qquad
    R_{j-1}=X_jM_j,
    \qquad
    Q_j=M_jZ_j,
\end{equation}
Eq.~\eqref{eq:space-uniformity-markov-factorization} applied at the
$j$th step gives the recursion
\begin{equation}
    \mathcal S(R_j)
    =
    \left\langle
        \mathcal S(R_{j-1}),\mathcal S(Q_j)
    \right\rangle .
    \label{eq:space-uniformity-stabilizer-recursion}
\end{equation}
The base case is $\mathcal S(R_1)=\mathcal S(Q_1)$.  If
\begin{equation}
    \mathcal S(R_{j-1})
    =
    \left\langle
        \mathcal S(Q_1),\ldots,\mathcal S(Q_{j-1})
    \right\rangle,
\end{equation}
then Eq.~\eqref{eq:space-uniformity-stabilizer-recursion} implies
\begin{equation}
    \mathcal S(R_j)
    =
    \left\langle
        \mathcal S(Q_1),\ldots,\mathcal S(Q_j)
    \right\rangle .
\end{equation}
Induction through $j=N$, together with $R_N=B$, yields
\begin{equation}
    \mathcal S(B)
    =
    \left\langle
        \mathcal S(Q_1),\ldots,\mathcal S(Q_N)
    \right\rangle .
\end{equation}
Thus the stabilizer group supported in $B$ has a generating set consisting
of stabilizers supported on the local patches $Q_j$.

\begingroup
\tikzset{
  pics/levin-wen-annulus/.style={code={
    \fill[blue!18] (35:0.52) arc (35:145:0.52) -- (145:1.35) arc (145:35:1.35) -- cycle;
    \fill[gray!22] (145:0.52) arc (145:215:0.52) -- (215:1.35) arc (215:145:1.35) -- cycle;
    \fill[red!18] (215:0.52) arc (215:325:0.52) -- (325:1.35) arc (325:215:1.35) -- cycle;
    \fill[gray!22] (325:0.52) arc (325:395:0.52) -- (395:1.35) arc (395:325:1.35) -- cycle;
    \draw[line width=1.8pt] (0,0) circle (1.35);
    \draw[line width=1.4pt] (0,0) circle (0.52);
    \foreach \angle in {35,145,215,325}{
      \draw[line width=1.2pt] (\angle:0.52) -- (\angle:1.35);
    }
    \node[font=\large] at (90:0.94) {$A$};
    \node[font=\large] at (180:0.94) {$B$};
    \node[font=\large] at (270:0.94) {$C$};
    \node[font=\large] at (0:0.94) {$B$};
  }},
  transport arrow/.style={
    line width=3.5pt,draw=blue!75!cyan,
    postaction={decorate},
    decoration={markings,mark=at position 0.55 with {\arrow{Stealth[length=4mm,width=3.5mm]}}},
    -{Stealth[length=5mm,width=4.5mm]}
  },
  placement label/.style={font=\large,fill=white,fill opacity=0.9,text opacity=1,rounded corners=2pt,inner sep=3pt}
}
\begin{figure}[!t]
\centering
\resizebox{0.33\textwidth}{!}{%
\begin{tikzpicture}[x=1cm,y=1cm,line cap=round,line join=round]
  \begin{scope}
    \clip[rounded corners=20pt] (0,0) rectangle (9,9);
    \fill[white] (0,0) rectangle (3.4,9);
    \shade[left color=white,right color=red!14] (3.4,0) rectangle (5.6,9);
    \fill[red!14] (5.6,0) rectangle (9,9);
  \end{scope}
  \draw[line width=2.4pt,rounded corners=20pt] (0,0) rectangle (9,9);
  \draw[transport arrow] (2.85,2.95) to[out=85,in=240] (5.95,6.05);
  \pic at (1.8,1.8) {levin-wen-annulus};
  \pic at (7.2,7.2) {levin-wen-annulus};
  \node[placement label] at (1.8,3.6) {$\mathcal G_0$};
  \node[placement label,fill=red!7] at (7.2,5.4) {$\mathcal G_m$};
  \node[font=\large,align=center] at (1.8,7.6) {clean region\\$Q$};
  \node[font=\large,align=center] at (7.1,1.3) {noise applied on\\$\overline Q$};
  \node[font=\large] at (4.5,-0.55) {(a)};
\end{tikzpicture}%
}
\hspace{0.06\textwidth}%
\resizebox{0.33\textwidth}{!}{%
\begin{tikzpicture}[x=1cm,y=1cm,line cap=round,line join=round]
  \filldraw[fill=red!14,line width=2.4pt,rounded corners=20pt] (0,0) rectangle (9,9);
  \draw[transport arrow] (5.95,6.05) to[out=240,in=85] (2.85,2.95);
  \pic at (1.8,1.8) {levin-wen-annulus};
  \pic at (7.2,7.2) {levin-wen-annulus};
  \node[placement label,fill=red!7] at (1.8,3.6) {$\mathcal G_0$};
  \node[placement label,fill=red!7] at (7.2,5.4) {$\mathcal G_m$};
  \node[font=\large,align=center] at (2.1,7.6) {fully noisy state\\$\rho=\mathcal N_Q(\sigma)$};
  \node[font=\large,align=center] at (7.1,1.3) {return to the\\original placement};
  \node[font=\large] at (4.5,-0.55) {(b)};
\end{tikzpicture}%
}
\caption{Spatial deformation of the Levin--Wen partition. (a) In $\sigma=\mathcal N_{\overline Q}(\rho_0)$, deform the annulus from $\mathcal G_0$ in the clean region $Q$ to $\mathcal G_m$ outside $Q$. (b) Apply the remaining noise on $Q$, leaving the reduced state on $\mathcal G_m$ unchanged, then reverse the deformation in $\rho$. Pink shading indicates applied noise; the gradient is schematic. }
\label{fig:levin-wen-transport}
\end{figure}
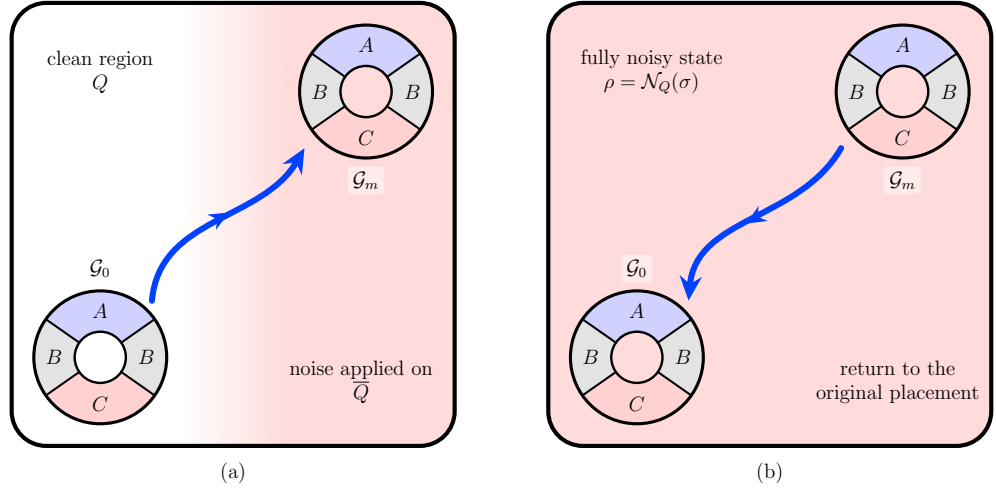
\endgroup

\section{Spatial deformation proof of Levin--Wen stability}
\label{sec:levin-wen-deformation}

We prove Corollary~\ref{cor:levin-wen-stability} using the spatial domain-wall construction of Ref.~\cite[Sec.~V]{yang2025topologicalmixedstatesphases}. For the Levin--Wen partition $\mathcal G=(A,B,C)$ in Fig.~\ref{fig:cmi-partitions}(c), write $\gamma_{\rm LW}(\omega;\mathcal G)=\frac12 I(A:C| B)_\omega$. Deformations will change the partition but not the state, and preserve the annular hole and the two components of $B$.

\subsection{Local deformation bounds}
Let $\delta(r)$ be the largest bound in Corollary~\ref{cor:space-uniformity-stability} for the space-uniformity partitions in Figure~\ref{fig:cmi-partitions}(b), all with separating width at least $r$. The bound includes partially applied onsite noise. For a simply connected patch $a$ outside $ABC$, enlarging $A$ to $Aa$ gives
\begin{equation}
\begin{aligned}
    I(Aa:C| B)_\omega-I(A:C| B)_\omega
    &=I(a:C| AB)_\omega\\
    &\leq I(a:BC| A)_\omega\leq\delta(r),
\end{aligned}
\label{eq:lw-endpoint-move}
\end{equation}
by the CMI chain rule, provided $(a,A,BC)$ has the space-uniformity geometry. An inward move can be obtained thorugh just data-processing inequality.

Transferring a patch $a\subset A$ into $B$, the chain rule gives
\begin{equation}
    I(A:C| B)_\omega-I(A\setminus a:C| Ba)_\omega
    =I(a:C| B)_\omega\leq\delta(r),
    \label{eq:lw-buffer-move}
\end{equation}
when $(a,B,C)$ has the space-uniformity geometry. Conversely, transferring $a\subset B$ into $A$ yields
\begin{equation}
    I(Aa:C| B\setminus a)_\omega-I(A:C| B)_\omega
    =I(a:C| B\setminus a)_\omega\leq\delta(r),
    \label{eq:lw-b-to-a-move}
\end{equation}
provided $(a,B\setminus a,C)$ has that geometry and the remaining buffer has width at least $r$. 
Both transfers leave $ABC$ unchanged; only the auxiliary CMI is required to be small. 
We call a deformation $r$-admissible if it is a finite sequence of the elementary moves above such that every auxiliary space-uniformity partition has separating width at least $r$, while the annular hole and the two components of $B$ are preserved. An $r$-admissible deformation can be constructed as shown in Ref.~\cite[Fig.~13(c) and Proposition~13]{yang2025topologicalmixedstatesphases}.
Each move satisfies the convenient bound
\begin{equation}
    \left|\gamma_{\rm LW}(\omega;\mathcal G_{j+1})-\gamma_{\rm LW}(\omega;\mathcal G_j)\right|
    \leq\delta(r),
    \label{eq:lw-local-deformation}
\end{equation}
including moves across a noisy--clean interface. The roles of $A$ and $C$ are symmetric, so the argument applies equally to $C$.

\subsection{Outward-and-return construction}

Choose a region $Q$ containing the original annulus $\mathcal G_0$ and an $r$-admissible sequence $\mathcal G_0\to\cdots\to\mathcal G_m$ ending entirely in $\overline Q:=\Lambda\setminus Q$. 
This requires sufficient room for the displaced annulus while maintaining the required separation width $r$.
Write $\mathcal N_Q:=\bigotimes_{i\in Q}\mathcal N_i$, acting trivially outside $Q$, and define $\mathcal N_{\overline Q}$ analogously.

First apply the noise on $\overline Q$, obtaining $\sigma:=\mathcal N_{\overline Q}(\rho_0)$. Since $\mathcal G_0\subset Q$, its reduced state is unchanged, so
\begin{align}
\gamma_{\rm LW}(\sigma;\mathcal G_0)=\gamma_{\rm LW}(\rho_0;\mathcal G_0).
\end{align}
Deform the annulus to $\mathcal G_m$ in the fixed state $\sigma$, as in Fig.~\ref{fig:levin-wen-transport}(a). Equation~\eqref{eq:lw-local-deformation} bounds the accumulated error by $m\delta(r)$.
Next apply the remaining noise on $Q$ to obtain $\rho=\mathcal N_Q(\sigma)=\mathcal N(\rho_0)$. Since $\mathcal G_m\subset\overline Q$, its reduced state is unchanged, and hence
\begin{align}
\gamma_{\rm LW}(\rho;\mathcal G_m)=\gamma_{\rm LW}(\sigma;\mathcal G_m).
\end{align}
Reverse the deformation in the fixed state $\rho$, as in Fig.~\ref{fig:levin-wen-transport}(b), incurring at most another $m\delta(r)$.

Thus the two changes of state occur at placements where the relevant reduced state is exactly unchanged, while every nontrivial comparison is between two spatial partitions of the same state. Combining the endpoint equalities gives
\begin{align}
\left|\gamma_{\rm LW}(\rho;\mathcal G_0)-\gamma_{\rm LW}(\rho_0;\mathcal G_0)\right|
\leq\left|\gamma_{\rm LW}(\rho;\mathcal G_0)-\gamma_{\rm LW}(\rho;\mathcal G_m)\right| +\left|\gamma_{\rm LW}(\sigma;\mathcal G_m)-\gamma_{\rm LW}(\sigma;\mathcal G_0)\right| 
\leq 2m\delta(r).
\end{align}
\label{eq}
Since $\gamma_{\rm LW}(\omega;\mathcal G_0)=\gamma_{\rm LW}(\omega)$, this proves Eq.~\eqref{eq:levin-wen-stability} with $C_{\rm def}=2m$.

\section{Proof of Corollary~\ref{cor:constant-depth-block-petz}}
\label{sec:proof-constant-depth-block-petz}

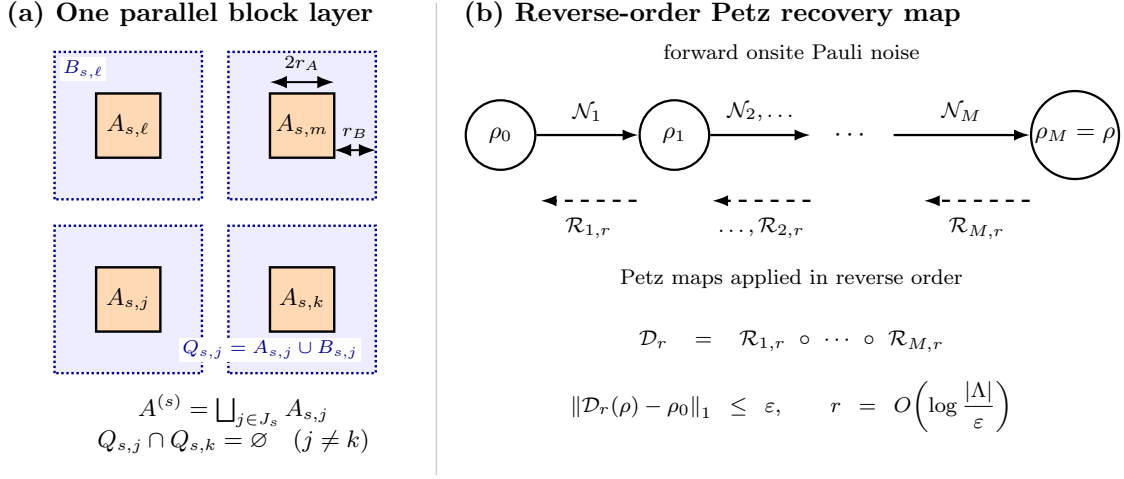
\begin{figure*}[t]
    \centering
    \begin{tikzpicture}[
        x=1cm,
        y=1cm,
        every node/.style={font=\small},
        paneltitle/.style={anchor=west,font=\bfseries\normalsize},
        recovery/.style={draw=blue!60!black, densely dotted, line width=0.9pt, fill=blue!7},
        core/.style={draw=black, line width=0.9pt, fill=orange!30},
        dim/.style={<->, >=Latex, line width=0.75pt},
        state/.style={
            circle,
            draw,
            line width=0.85pt,
            minimum size=0.92cm,
            inner sep=0pt,
            font=\small
        },
        forward/.style={-{Latex[length=2.1mm,width=1.5mm]}, line width=0.85pt},
        reverse/.style={-{Latex[length=2.1mm,width=1.5mm]}, dashed, line width=0.85pt},
        note/.style={font=\footnotesize, align=center},
        eqnote/.style={font=\scriptsize, align=center, text width=4.8cm},
        arrowlabel/.style={font=\footnotesize, fill=white, inner sep=1pt}
    ]
    \path[use as bounding box] (-0.35,-1.10) rectangle (14.45,5.95);

    \begin{scope}
        \node[paneltitle]
            at (-0.75,5.55)
            {(a) One parallel block layer};

        \draw[recovery] (0,0.80) rectangle (1.95,2.75);
        \draw[recovery] (2.30,0.80) rectangle (4.25,2.75);
        \draw[recovery] (0,3.10) rectangle (1.95,5.05);
        \draw[recovery] (2.30,3.10) rectangle (4.25,5.05);

        \draw[core] (0.55,1.35) rectangle (1.40,2.20);
        \draw[core] (2.85,1.35) rectangle (3.70,2.20);
        \draw[core] (0.55,3.65) rectangle (1.40,4.50);
        \draw[core] (2.85,3.65) rectangle (3.70,4.50);

        \node at (0.975,1.775) {$A_{s,j}$};
        \node at (3.275,1.775) {$A_{s,k}$};
        \node at (0.975,4.075) {$A_{s,\ell}$};
        \node at (3.275,4.075) {$A_{s,m}$};

        \node[font=\scriptsize, blue!60!black, fill=white, inner sep=1pt]
            at (0.38,4.82)
            {$B_{s,\ell}$};

        \node[
            blue!60!black,
            anchor=south east,
            align=right,
            font=\scriptsize,
            fill=white,
            inner sep=1pt
        ]
        at (4.06,0.95)
        {$Q_{s,j}=A_{s,j}\cup B_{s,j}$};

        \draw[dim]
            (2.85,4.65)
            --
            node[above=1pt,font=\scriptsize] {$2r_A$}
            (3.70,4.65);

        \draw[dim]
            (3.70,3.75)
            --
            node[above=1pt,font=\scriptsize] {$r_B$}
            (4.25,3.75);

        \node[align=center] at (2.37, 0.13) {
        $A^{(s)}=\bigsqcup_{j\in J_s}A_{s,j}$ \\
        $Q_{s,j}\cap Q_{s,k}=\varnothing\quad(j\neq k)$
        };

    \end{scope}

    \draw[black!25] (5.05,-0.55) -- (5.05,5.70);

    \begin{scope}[shift={(5.65,0)}]
        \node[paneltitle]
            at (-0.35,5.55)
            {(b) Reverse-order Petz recovery map};

        \node[note]
            at (4.10,5.05)
            {forward onsite Pauli noise};

        \node[state] (rho0) at (0.25,3.95)
            {$\rho_0$};

        \node[state] (rho1) at (2.55,3.95)
            {$\rho_1$};

        \node[font=\normalsize] at (4.90,3.95)
            {$\cdots$};

        \node[state, text width=1.15cm, align=center] (rhoM) at (7.85,3.95)
            {$\rho_M =\rho$};

        \draw[forward]
            (rho0.east)
            --
            node[arrowlabel, above=4pt]
            {$\mathcal{N}_1$}
            (rho1.west);

        \draw[forward]
            (rho1.east)
            --
            node[arrowlabel, above=4pt]
            {$\mathcal{N}_2,\ldots$}
            (4.35,3.95);

        \draw[forward]
            (5.45,3.95)
            --
            node[arrowlabel, above=4pt]
            {$\mathcal{N}_M$}
            (rhoM.west);

        \draw[reverse]
            (2.05,3.08)
            --
            node[arrowlabel, below=5pt]
            {$\mathcal{R}_{1,r}$}
            (0.78,3.08);

        \draw[reverse]
            (4.35,3.08)
            --
            node[arrowlabel, below=5pt]
            {$\ldots,\mathcal{R}_{2,r}$}
            (3.05,3.08);

        \draw[reverse]
            (7.25,3.08)
            --
            node[arrowlabel, below=5pt]
            {$\mathcal{R}_{M,r}$}
            (5.85,3.08);

        \node[note]
            at (4.10,2.05)
            {Petz maps applied in reverse order};

        \node[font=\footnotesize, align=center, text width=8.0cm]
            at (4.10,1.22)
            {$\displaystyle
            \mathcal{D}_r
            =
            \mathcal{R}_{1,r}
            \circ\cdots\circ
            \mathcal{R}_{M,r}
            $};

        \node[font=\footnotesize, align=center, text width=8.0cm]
            at (4.10,0.36)
            {$\displaystyle
            \left\|
                \mathcal{D}_r(\rho)-\rho_0
            \right\|_1
            \leq \varepsilon,
            \qquad
            r=
            O\!\left(
                \log\frac{|\Lambda|}{\varepsilon}
            \right)
            $};
    \end{scope}

    \end{tikzpicture}

    \caption{
    Block construction and reverse-order implementation of the quasi-local Petz decoder.
    \textbf{(a)}
    At a fixed layer $s$, where $J_s$ indexes the blocks in that layer, the noise support is decomposed as a (disjoint) union of non-overlapping blocks: 
    $A^{(s)}=\bigsqcup_{j\in J_s}A_{s,j}$.
    Each core block $A_{s,j}$ has radius $r_A$ and is surrounded by a buffer $B_{s,j}$ of width $r_B$, with $r_A=\Theta(r_B)$.
    The recovery region is $Q_{s,j}=A_{s,j}\cup B_{s,j}$, and the blocks are chosen so that $Q_{s,j}\cap Q_{s,k}=\varnothing$ for $j\neq k$ within the same layer.
    Hence, the corresponding Petz recovery maps can be applied in parallel, giving constant depth per layer.
    \textbf{(b)}
    The onsite Pauli noise channels generate the sequence 
    $\rho_0,\rho_1,\ldots,\rho_M=\rho$.
    The decoder applies the quasi-local Petz recovery maps in the reverse order, $\mathcal{D}_r = \mathcal{R}_{1,r}\circ\cdots\circ\mathcal{R}_{M,r}$, and achieves
    $\|\mathcal{D}_r(\rho)-\rho_0\|_1\leq\varepsilon$
    for recovery radius
    $r=O(\log(|\Lambda|/\varepsilon))$.
    }
    \label{fig:constant-depth-block-petz}
\end{figure*}

\begin{appendixcorollary}[Formal statement of
Corollary~\ref{cor:constant-depth-block-petz}]
\label{cor:formal-constant-depth-block-petz}
Under the assumptions of Corollary~\ref{cor:constant-depth-block-petz},
let $\rho_0$ be the normalized code state, write the onsite Pauli noise as
$\mathcal N=\bigotimes_{i\in\Lambda}\mathcal N_i$, and set
$\rho=\mathcal N(\rho_0)$.
The block construction and reverse-order decoder are illustrated in
Fig.~\ref{fig:constant-depth-block-petz}. Let $\Lambda$ be a $D$-dimensional lattice and write $n:=|\Lambda|$.  Fix a length scale $R$, and choose block and buffer radii
\begin{align}
    r_A=\alpha R, \qquad  r_B=R,
\end{align}
where $0<\alpha<1$. Partition the lattice into blocks
$\{A_{s,j}\}$ with
\begin{align}
    s=1,\ldots,\chi, \quad \chi = O\!\left(\frac{r_B}{r_A}\right)^D
    =
    O(1)
\end{align}
Let $J_s$ denote the index set of blocks belonging to family $s$. The families are chosen such that
\begin{align}
    \Lambda
    =
    \bigsqcup_{s=1}^{\chi}
    \bigsqcup_{j\in J_s} A_{s,j},
    \quad
    |A_{s,j}|=\Theta(R^D),
\end{align}
and define
\begin{align}
    Q_{s,j}
    :=
    N_{r_B}(A_{s,j}),
    \quad
    B_{s,j}
    :=
    Q_{s,j}\setminus A_{s,j},
    \quad
    C_{s,j}
    :=
    \Lambda\setminus Q_{s,j},
\end{align}
Since the enlarged regions in each family are pairwise disjoint:
\begin{align}
    Q_{s,j}\cap Q_{s,k}=\varnothing,
    \qquad
    j\neq k.
    \label{eq:disjoint-enlarged-blocks}
\end{align}
All recovery maps within one family have disjoint support and
can be applied in parallel.
For each family, define
$
    A^{(s)}
    :=
    \bigsqcup_{j\in J_s} A_{s,j}
$
and group the product Pauli noise as
\begin{align}
    \mathcal{N}
    =
    \mathcal{C}_{\chi}\circ\cdots\circ\mathcal{C}_1,
    \quad
    \mathcal{C}_s
    :=
    \bigotimes_{j\in J_s}\mathcal{N}_{s,j},
    \quad
    \mathcal{N}_{s,j}
    :=
    \bigotimes_{i\in A_{s,j}}\mathcal{N}_i .
\end{align}

Then there exist twirled Petz maps
$\mathcal{R}_{s,j,R}$ supported on $Q_{s,j}$ such that
\begin{equation}
    \widetilde{\mathcal{R}}_{s,R}
    :=
    \bigotimes_{j\in J_s}\mathcal{R}_{s,j,R}
    \label{eq:parallel-petz-layer}
\end{equation}
And the reverse-order Petz recovery map
\begin{equation}
    \mathcal{D}_R
    :=
    \widetilde{\mathcal{R}}_{1,R}
    \circ\cdots\circ
    \widetilde{\mathcal{R}}_{\chi,R}
    \label{eq:block-petz-decoder}
\end{equation}
satisfies
\begin{equation}
    \left\|
        \mathcal{D}_R(\rho)-\rho_0
    \right\|_1
    \leq \varepsilon
    \quad \text{with} \quad
     R
    =
    \Theta\!\left(
        \log\frac{n}{\varepsilon}
    \right)
    \label{eq:block-petz-global-error}
\end{equation}
so each Petz map acts on
\begin{align}
    |Q_{s,j}|
    =
    O(R^D)
    =
    O\!\left(
        \log^D\frac{n}{\varepsilon}
    \right)
\end{align}
qudits, and  the depth of the quasi-local block is
\begin{equation}
    \operatorname{depth}_{\mathrm{Petz}}(\mathcal{D}_R)
    \leq \chi
    =
    O(1).
    \label{eq:constant-petz-depth}
\end{equation}
\end{appendixcorollary}

\begin{proof}
Let $\sigma_s$ denote the global state immediately before applying
$\mathcal C_s$.  Fix an arbitrary ordering of the blocks in $J_s$, apply
the channels $\mathcal N_{s,j}$ sequentially in this order, and let
$\sigma_{s,j}$ denote the global intermediate state immediately before
applying $\mathcal N_{s,j}$.  The approximate-sufficiency theorem of
Ref.~\cite{Junge2015} states that the twirled Petz map
$\mathcal P_{\mathcal E,\omega}$ satisfies, for density matrices
$\varrho$ and $\omega$ and a quantum channel $\mathcal E$,
\begin{equation*}
    D(\varrho\|\omega)
    -D(\mathcal E(\varrho)\|\mathcal E(\omega))
    \geq
    -2\log_2 F\!\left(
        \varrho,
        \mathcal P_{\mathcal E,\omega}\mathcal E(\varrho)
    \right).
\end{equation*}
Take $\varrho=\sigma_{s,j}$,
$\omega=(\sigma_{s,j})_{A_{s,j}B_{s,j}}\otimes
(\sigma_{s,j})_{C_{s,j}}$, and $\mathcal E=\mathcal N_{s,j}$.  Since
$\mathcal E$ acts only on $A_{s,j}$, we have
\begin{align*}
    I(A_{s,j}B_{s,j}:C_{s,j})_{\sigma_{s,j}}
    -I(A_{s,j}B_{s,j}:C_{s,j})_{\mathcal N_{s,j}(\sigma_{s,j})}
    ={}&I(A_{s,j}:C_{s,j}| B_{s,j})_{\sigma_{s,j}}-I(A_{s,j}:C_{s,j}| B_{s,j})_{\mathcal N_{s,j}(\sigma_{s,j})}.
\end{align*}
Moreover, writing $T(\varrho,\tau):=\frac12\|\varrho-\tau\|_1$, the
Fuchs--van de Graaf inequality
$T(\varrho,\tau)\leq\sqrt{1-F(\varrho,\tau)^2}$
~\cite{fuchs1998cryptographicdistinguishabilitymeasuresquantum}
implies
\begin{equation*}
    -2\log_2F(\varrho,\tau)
    \geq
    -\log_2\!\left(1-T(\varrho,\tau)^2\right)
    \geq
    \frac{1}{4\ln2}\|\varrho-\tau\|_1^2.
\end{equation*}
Denoting the resulting recovery
map on $A_{s,j}B_{s,j}=Q_{s,j}$ by $\mathcal R_{s,j,R}$ gives
\begin{align}
    \frac{1}{4\ln 2}
    \left\|
        \mathcal{R}_{s,j,R}
        \mathcal{N}_{s,j}(\sigma_{s,j})
        -
        \sigma_{s,j}
    \right\|_1^2
    &\leq
    I(A_{s,j}:C_{s,j}| B_{s,j})_{\sigma_{s,j}}
    -
    I(A_{s,j}:C_{s,j}| B_{s,j})_{\mathcal N_{s,j}(\sigma_{s,j})}
    \notag\\
    &\leq
    I(A_{s,j}:C_{s,j}| B_{s,j})_{\sigma_{s,j}}.
    \label{eq:block-local-petz}
\end{align}
Every $\sigma_{s,j}$ is obtained from $\rho_0$ by applying a subset of the
onsite Pauli channels $\mathcal N_i$, with the identity channel on all
remaining sites.  It therefore satisfies the same noise threshold as the
full channel, so Theorem~\ref{thm:cmi-stability} applies and gives
\[
    I(A_{s,j}:C_{s,j}| B_{s,j})_{\sigma_{s,j}}
    \leq
    K (|\partial B_{s,j,A}|+|\partial B_{s,j,C}|)^2e^{-c_0R}.
\]
Because both the core radius and the buffer width are $\Theta(R)$,
bounded lattice geometry gives
$
    (|\partial B_{s,j,A}|+|\partial B_{s,j,C}|)^2=O(R^{2D-2}).
$
Hence every block recovery error obeys
\begin{equation}
    \left\|
        \mathcal{R}_{s,j,R}
        \mathcal{N}_{s,j}(\sigma_{s,j})
        -
        \sigma_{s,j}
    \right\|_1
    =
    O\!\left(
        R^{D-1}e^{-c_0R/2}
    \right).
    \label{eq:block-single-error}
\end{equation}

The total number of blocks is $O\!\left(\frac{n}{R^D}\right)$.
Order all block channels consistently with the forward evolution and the
recovery maps in reverse order; inserting the corresponding intermediate
states $\sigma_{s,j}$ and using trace-norm contractivity under CPTP maps
yields
\begin{align}
    \left\|
        \mathcal{D}_R(\rho)-\rho_0
    \right\|_1
    &\leq
    \sum_{s,j}
    \left\|
        \mathcal{R}_{s,j,R}
        \mathcal{N}_{s,j}(\sigma_{s,j})
        -
        \sigma_{s,j}
    \right\|_1
    \notag\\
    &=
    O\!\left(
        nR^{-1}e^{-c_0R/2}
    \right).
\end{align}
Taking $R=\Theta(\log(n/\varepsilon))$ makes the expression at
most $\varepsilon$. Finally, Eq.~\eqref{eq:disjoint-enlarged-blocks} implies that all
maps $\mathcal{R}_{s,j,R}$ with fixed $s$ have disjoint support and
may be applied simultaneously. Thus each $\widetilde{\mathcal{R}}_{s,R}$ is one recovery layer and the total block-map depth is at most $\chi=O(1)$.
\end{proof}
Finally, we believe that the same recovery map construction can be adapted to the geometrically $k$-local noise easily.

\end{document}